%% file: Formatting-Instructions-LaTeX-2026.tex
\documentclass[letterpaper]{article} 
\usepackage{aaai2026}  
\usepackage{times}  
\usepackage{helvet}  
\usepackage{courier}  
\usepackage[hyphens]{url}  
\usepackage{graphicx} 
\usepackage{natbib}  
\usepackage{caption} 
\usepackage[ruled,linesnumbered,vlined,noend]{algorithm2e}

\usepackage{subfigure}
\usepackage{amsthm}
\usepackage{color}
\usepackage{tikz}
\usetikzlibrary{positioning, calc}
\usepackage{tikz-qtree}
\usepackage{pgfplots}
\usepackage{pgfplotstable}
\usepackage[hypcap=true]{subcaption}
\usetikzlibrary{patterns}
\usepackage{tabularx}
\usepackage{makecell}
\usepackage{multirow}
\usepackage{booktabs}
\usepackage{bm}
\usepackage[table,xcdraw]{xcolor}
\usepackage{amsmath}
\usepackage{caption} 
\usepackage{enumitem}
\usepackage{array}     
\usepackage{utfsym} 

\definecolor{mycolor2}{RGB}{164,224,187} 
\definecolor{mycolor1}{RGB}{53,122,162} 

\newcommand{\colorcell}[1]{%
    \cellcolor{mycolor1!#1!mycolor2}%
    \ifnum#1>49%
        {\Large \textcolor{white}{{#1}}}
    \else
        {\Large \textcolor{black}{{#1}}}
    \fi
}

\newtheorem{definition}{Definition} 
\newtheorem{lemma}{Lemma}

\usepackage[most]{tcolorbox}
\usepackage{float}
\usepackage{xspace}
\tcbset{
  aibox/.style={
    width=474.18663pt,
    top=10pt,
    colback=white,
    colframe=black,
    colbacktitle=black,
    enhanced,
    center,
    attach boxed title to top left={yshift=-0.1in,xshift=0.15in},
    boxed title style={boxrule=0pt,colframe=white,},
  }
}
\newtcolorbox{AIbox}[2][]{aibox,title=#2,#1}

\usepackage{listings}
\tcbset{
  ragstyle/.style={
    colback=white,
    colframe=black,
    boxrule=0.8pt,
    toprule=1.5pt,
    bottomrule=1.5pt,
    arc=0mm,
    left=2mm,
    right=2mm,
    top=2mm,
    bottom=2mm,
    width=\textwidth,
    enlarge left by=0mm,
    before skip=1em,
    after skip=1em
  }
}

\usepackage{newfloat}
\usepackage{listings}
\DeclareCaptionStyle{ruled}{labelfont=normalfont,labelsep=colon,strut=off} 
\floatstyle{ruled}
\newfloat{listing}{tb}{lst}{}
\floatname{listing}{Listing}
\title{ArchRAG: Attributed Community-based Hierarchical \\ Retrieval-Augmented Generation}
\author{
    Shu Wang\textsuperscript{\rm 1},
    Yixiang Fang\textsuperscript{\rm 1}\protect\thanks{Corresponding author.},
    Yingli Zhou\textsuperscript{\rm 1},
    Xilin Liu\textsuperscript{\rm 2},
    Yuchi Ma\textsuperscript{\rm 2}
}
\affiliations{
    \textsuperscript{\rm 1}The Chinese University of Hong Kong, Shenzhen\\
    \textsuperscript{\rm 2}Huawei Cloud Computing Technologies CO., LTD.\\
    \{shuwang3, yinglizhou\}@link.cuhk.edu.cn, 
    fangyixiang@cuhk.edu.cn,
    \{liuxilin3, mayuchi1\}@huawei.com
}

\usepackage{bibentry}

\begin{document}

\maketitle

\input{Section/0_Abstract}

\begin{links}
    \link{Code}{https://github.com/sam234990/ArchRAG}
    \link{Extended version}{https://arxiv.org/abs/2502.09891}
\end{links}


\input{Section/1_Introduction}

\input{Section/7_RelatedWork}
\input{Section/4_Framework}

\input{Section/6_Experiment}
\input{Section/8_Conclusion}

\section{Acknowledgments}
This work was supported in part by the 1+1+1 CUHK-CUHK(SZ)-GDSTC Joint Collaboration Fund under Grant 2025A0505000045, Guangdong Provincial Key Laboratory of Mathematical Foundations for Artificial Intelligence (2023B1212010 001), Shenzhen Research Institute of Big Data under Grant SIF20240002, and Huawei Collaboration Fund under Grant TC20240920019.

{
\small
\bibliography{aaai2026, references}
}

\input{Section/9_Appendix}

\end{document}

%% file: Section/0_Abstract.tex
\begin{abstract}
Retrieval-Augmented Generation (RAG) has proven effective in integrating external knowledge into large language models (LLMs) for solving question-answer (QA) tasks.
The state-of-the-art RAG approaches often use the graph data as the external data since they capture the rich semantic information and link relationships between entities.
However, existing graph-based RAG approaches cannot accurately identify the relevant information from the graph and also consume large numbers of tokens in the online retrieval process.
To address these issues, we introduce a novel graph-based RAG approach, called \underline{A}tt\underline{r}ibuted \underline{C}ommunity-based \underline{H}ierarchical \underline{RAG} (ArchRAG), by augmenting the question using attributed communities, and also introducing a novel LLM-based hierarchical clustering method.
To retrieve the most relevant information from the graph for the question, we build a novel hierarchical index structure for the attributed communities and develop an effective online retrieval method.
Experimental results demonstrate that ArchRAG outperforms existing methods in both accuracy and token cost.
%
\end{abstract}

%
%

%% file: Section/1_Introduction.tex
\section{Introduction}

Retrieval-Augmented Generation (RAG) has emerged as a core approach for enhancing large language models (LLMs) by enabling access to domain-specific and real-time updated knowledge beyond their pre-training corpus~\shortcite{fan2024survey,gao2023retrieval,hu2024rag,huang2024survey,wu2024retrieval,yu2024evaluation,zhao2024retrieval}.
By improving the trustworthiness and interpretability of LLMs, RAG has been widely adopted across a broad range of applications~\shortcite{liu2024survey,zheng2024large,nie2024survey,li2023large,ghimire2024generative,wang2024large}.
Current state-of-the-art RAG approaches often use graph-structured data as external knowledge, due to its ability to capture the rich semantics and relationships.
Given a question $Q$, the key idea of graph-based RAG is to retrieve relevant information (e.g., nodes, subgraphs, or textual information) from the graph, incorporate them with $Q$ as the prompt, and feed them into the LLM, as illustrated in Figure \ref{fig:dual-tasks}.
Several recent methods~\shortcite{edge2024local,sarthi2024raptor,guo2024lightrag,wu2024medical,wang2024knowledge,li2024dalk,gutierrez2024hipporag} address two common question answering (QA) tasks: abstract questions, which require reasoning over high-level themes (e.g., “What are the potential impacts of LLMs on education?”), and specific questions, which focus on entity-centric factual details (e.g., “Who won the Turing Award in 2024?”).

%


\begin{figure}[t]
    \centering
    \includegraphics[width=\linewidth]{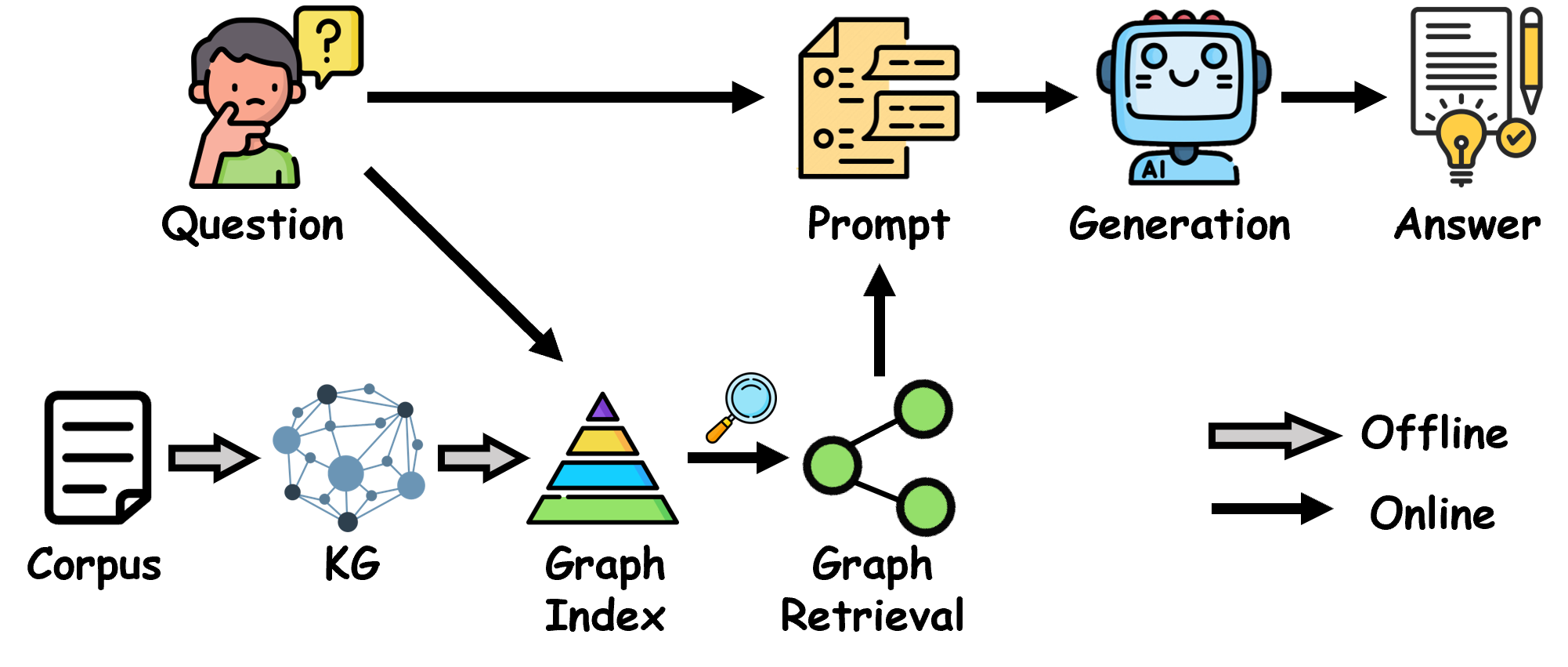}
    \caption{The general workflow of graph-based RAG, which retrieves relevant information (e.g., nodes, subgraphs, or textual information) to facilitate the LLM generation.}
    \label{fig:dual-tasks}
\end{figure}

%
In the past year, a surge of graph-based RAG methods~\shortcite{huang2025ket,wang2024knowledge,li2024dalk,gutierrez2024hipporag,ma2024think,sun2023think} has emerged, each proposing different retrieval strategies to extract detailed information for response generation.
Among them, GraphRAG~\shortcite{edge2024local}, proposed by Microsoft, is the most prominent and the first to leverage community summarization for abstract QA.
It builds a knowledge graph (KG) from the external corpus, detects communities using Leiden~\shortcite{traag2019louvain}, and generates a summary for each community using LLMs. 
For abstract questions that require high-level information, it adopts a Global Search approach, traversing all communities and using LLMs to retrieve the most relevant summaries.
In contrast, for specific questions, it employs a Local Search method to retrieve entities, relevant text chunks, and low-level communities, providing the multi-hop detailed information for accurate answers.

Although some methods claim that GraphRAG underperforms and is difficult to apply in practice~\shortcite{ma2024think,zhang2024sirerag}, our re-examination shows that it is primarily constrained by the following limitations (\textbf{L}):
\textbf{L1.}
\textit{Low community quality}: GraphRAG uses the Leiden~\shortcite{traag2019louvain} algorithm to detect communities, but this approach relies solely on graph structure and ignores the rich semantics of nodes and edges.
As a result, the detected communities often consist of different themes, which leads to the poor quality of community summaries and further decreases its performance.
\textbf{L2.}
\textit{Limited compatibility}:
While GraphRAG employs Global and Local Search strategies, each retrieves graph elements at only one granularity, making it inadequate for simultaneously addressing both abstract and specific questions and limiting its applicability in real-world open-ended scenarios.
\textbf{L3.} \textit{High generation cost}:
Although GraphRAG performs well on abstract questions, analyzing all communities with LLMs is both time- and token-consuming.
For example, GraphRAG detects 2,984 communities in the Multihop-RAG~\shortcite{tang2024multihop} dataset, and answering just 100 questions incurs a cost of approximately \$650 and 106 million tokens\footnote{The cost of GPT-4o is \$10/M tokens for output and \$2.50/M tokens for input (for details, please refer to OpenAI pricing).}, which is an impractical overhead.


To tackle the above limitations of GraphRAG, in this paper, we propose a novel graph-based RAG approach, called \textbf{\underline{A}}tt\textbf{\underline{r}}ibuted \textbf{\underline{C}}ommunity-based \textbf{\underline{H}}ierarchical \textbf{\underline{RAG}} (ArchRAG).
ArchRAG leverages attributed communities (ACs) and introduces an efficient hierarchical retrieval strategy to adaptively support both abstract and specific questions.
To mitigate \textbf{L1}, we detect high-quality ACs by exploiting both links and the attributes of nodes, ensuring that each AC comprises nodes that are not only densely connected but also share similar semantic themes~\shortcite{zhou2009graph}.
We further propose a novel LLM-based iterative framework for hierarchical AC detection, which can incorporate any existing community detection methods~\shortcite{zhou2009graph,traag2019louvain,von2007tutorial,xu2007scan,grover2016node2vec}.
In each iteration, we detect ACs based on both attribute similarity and connectivity, summarize each AC using an LLM, and construct a higher-level graph by treating each AC as a node, connecting pairs with similar summaries.
By iterating the above steps multiple times, we obtain a set of ACs that can be organized into a hierarchical tree structure.

To effectively address \textbf{L2}, we organize all ACs and entities into a hierarchical index and retrieve relevant elements from all levels to support both abstract and specific questions.
Entities offer fine-grained details, while LLM-generated AC summaries capture relational structures and provide high-level condensed overviews~\shortcite{sarthi2024raptor,angelidis2018summarizing}, making them suitable for both multi-hop reasoning and abstract insight extraction.
To support efficient retrieval across levels, we propose C-HNSW (Community-based HNSW), a novel hierarchical index inspired by the HNSW algorithm~\shortcite{malkov2018efficient} for approximate nearest neighbor (ANN) search.
%

To mitigate the high generation cost caused by traversing all communities (\textbf{L3}), we propose a hierarchical search with adaptive filtering to efficiently select the most relevant ACs and entities while maintaining performance.
Specifically, we design an efficient hierarchical retrieval algorithm over the proposed C-HNSW index, which supports top-$k$ nearest neighbor search across multiple levels, thereby facilitating access to multi-level relevant information.
Furthermore, the adaptive filtering mechanism identifies the most informative results at each level, making the retrieved information complementary.

We have extensively evaluated ArchRAG on real-world datasets, and the results show that it consistently outperforms existing methods in both abstract and specific QA tasks.
ArchRAG achieves a 10\% higher accuracy than state-of-the-art graph-based RAG methods on specific questions and shows notable gains on abstract QA.
%
Moreover, ArchRAG is very token-efficient, saving up to 250 times the token usage compared to GraphRAG~\shortcite{edge2024local}.

In summary, our main contributions are as follows:
\begin{itemize}
    \item We present a novel graph-based RAG approach by using ACs that are organized hierarchically and detected by an LLM-based hierarchical clustering method.
    
    \item To index ACs, we propose a novel hierarchical index structure called C-HNSW and also develop an efficient online retrieval method.
    
    \item Extensive experiments show that ArchRAG is both highly effective and efficient, and achieves state-of-the-art performance on both abstract and specific QA tasks.
\end{itemize}

\begin{figure*}
    \centering
    \includegraphics[width=1\linewidth]{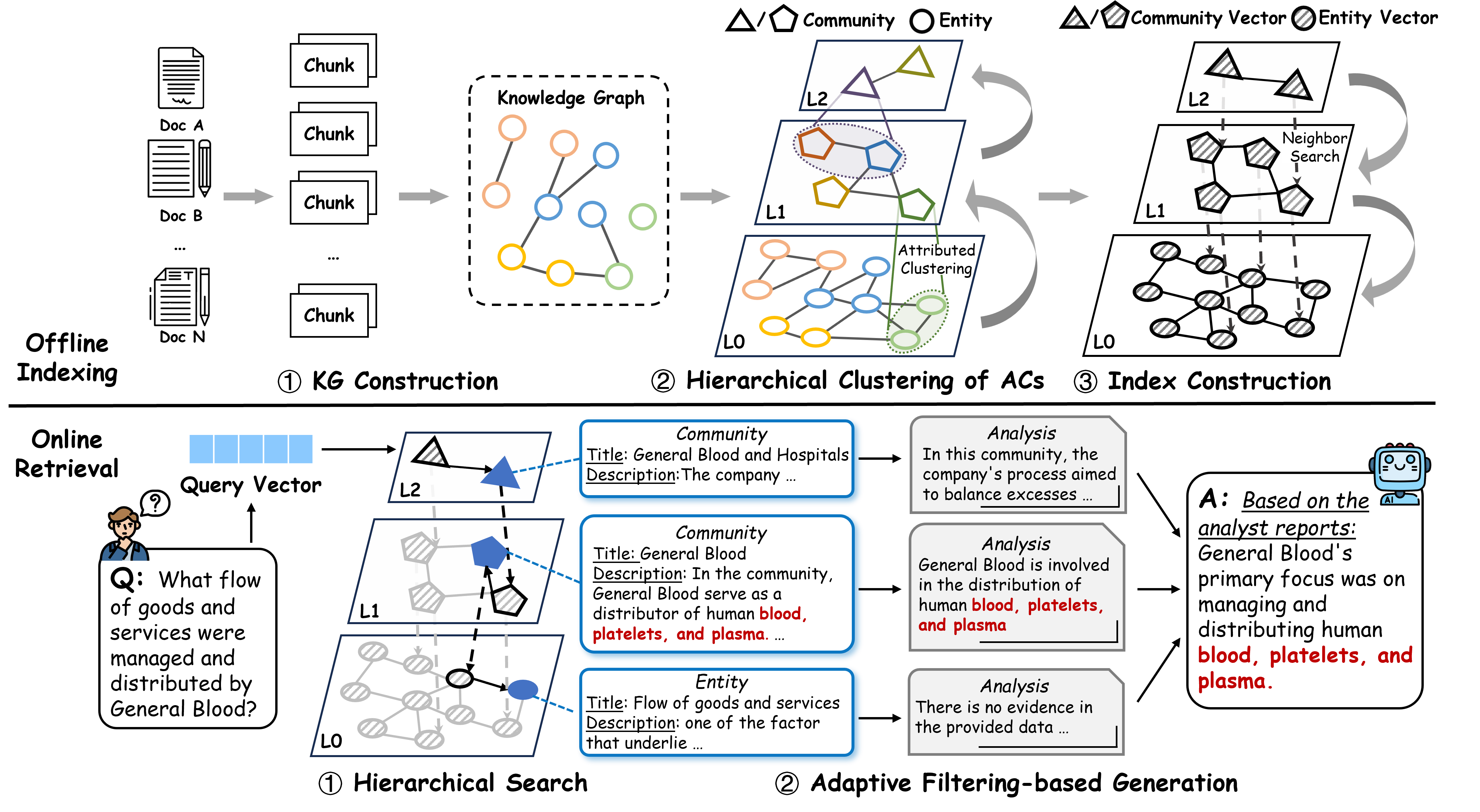}
    \caption{ArchRAG consists of two phases: offline indexing and online retrieval. For the online retrieval phase, we show an example of using ArchRAG to answer a question in the HotpotQA dataset.}
    \label{fig:overview}
\end{figure*}

%% file: Section/7_RelatedWork.tex
\section{Related Work}
\label{sec:related-work}

In this section, we review the related works, including Retrieval-Augmentation-Generation (RAG) approaches, and LLMs for graph mining and learning.

$\bullet$ \textbf{RAG approaches.} RAG has been proven
to excel in many tasks, including open-ended question answering \shortcite{jeong2024adaptive,siriwardhana2023improving}, programming context \shortcite{chen2024auto,chen2023haipipe}, SQL rewrite~\shortcite{lillm,sun2024r}, and data cleaning~\shortcite{naeem2024retclean,narayan2022can,qian2024unidm}.
The naive RAG technique relies on retrieving query-relevant contexts from external knowledge bases to mitigate the ``hallucination'' of LLMs.
Recently, many RAG approaches~\shortcite{huang2025ket,guo2024lightrag,gutierrez2024hipporag,sarthi2024raptor,wang2024knowledge,he2024g,ma2024think,li2024dalk} have adopted graph structures to organize the information and relationships within documents, leading to improved performance.
For more details, please refer to the recent survey of graph-based RAG methods~\shortcite{peng2024graph}.

$\bullet$ \textbf{LLM for graph mining. } 
Recent advances in LLMs have offered opportunities to leverage LLMs in graph mining. 
These include using LLMs for KG construction~\shortcite{zhu2024llms}, addressing complex graph mining tasks~\shortcite{chen2024graphwiz,tang2024grapharena,cao2024graphinsight,zhang2024gcoder}, and employing KG to enhance the LLM reasoning~\shortcite{wang2023knowledge,wang2024knowledge,luo2023reasoning,sun2023think,mavromatis2024gnn,jiang2023structgpt,huang2025ket}.
%
For instance, RoG \shortcite{luo2023reasoning} proposes a planning-retrieval-reasoning framework that retrieves reasoning paths from KGs to guide LLMs conducting faithful reasoning.
%
%
StructGPT \shortcite{jiang2023structgpt} and ToG~\shortcite{sun2023think} treat LLMs as agents that interact with KGs to find reasoning paths leading to the correct answers.
%


%% file: Section/4_Framework.tex
\section{Our Approach ArchRAG}
\label{sec:ArchRAG}

We begin by presenting the overall workflow and design rationale of ArchRAG, followed by detailed descriptions of each component.
As illustrated in Figure~\ref{fig:overview}, our proposed ArchRAG consists of two phases. 
In the {\it offline indexing} phase, ArchRAG first constructs a KG from the corpus, then detects ACs by a novel LLM-based hierarchical clustering method, and finally builds the C-HNSW index.
During the {\it online retrieval} phase, ArchRAG first converts the question into a query vector, then retrieves relevant information from the C-HNSW index, and finally generates answers through an adaptive filtering-based generation process

Our ArchRAG detects ACs by exploiting both links and attributes, and organizes ACs and entities into a novel hierarchical index, C-HNSW, yielding the following advantages:
1) Each AC group densely connected entities with shared themes and a high-quality summary.
2) The hierarchical structure captures multiple levels of abstraction: lower-level entities and communities encode detailed KG knowledge, while higher-level communities provide global context, enabling ArchRAG to address questions at varying granularity.
and 3) The C-HNSW index efficiently retrieves relevant information across levels, supporting fast and accurate responses to both abstract and specific questions.

\input{Section/Framework/index}

\input{Section/Framework/online}

%% file: Section/Framework/index.tex
\subsection{Offline Indexing}


\subsubsection{KG construction.}
ArchRAG builds a KG by prompting the LLM to extract entities and relations from each chunk of the text corpus $D$.
Specifically, all text contexts are segmented into chunks based on specified chunk length, enabling the LLM to extract entities and relations from each chunk using in-context learning~\shortcite{brown2020language}, thus forming subgraphs.
These subgraphs are then merged, with entities and relations that appear repeatedly across multiple subgraphs being consolidated by the LLM to generate a complete description.
Finally, we get a KG, denoted by $G(V,E)$, where $V$ and $E$ are sets of vertices and edges, respectively, and each vertex and edge is associated with textual attributes.

\subsubsection{LLM-based hierarchical clustering. } 
\label{sec:llm_cluster}
%
%
We propose an iterative LLM-based hierarchical clustering framework that supports arbitrary graph augmentation (e.g., KNN connections and CODICIL~\shortcite{CODICIL2013efficient}) and graph clustering algorithms (e.g., weighted Leiden \shortcite{traag2019louvain}, weighted spectral clustering, and SCAN \shortcite{xu2007scan}).
Specifically, we propose to augment the KG by linking entities if their attribute similarities are larger than a threshold, and then associate each pair of linked entities with a weight denoting their attribute similarity value.
Next, we generate the ACs using any given graph clustering algorithm.
In this way, both node attributes and structural links are jointly considered during community detection.

Algorithm \ref{alg:LHC} shows the above iterative clustering process.
Given a graph augmentation method $Aug$, clustering algorithm $GCluster$ and stopping condition $T$, we perform the following steps in each iteration:
(1) augmenting the graph (line 3);
(2) computing the edge weights (lines 4-5);
(3) clustering the augmented graph (line 6);
(4) generating a summary for each community using LLM (line 7);
and (5) building a new attributed graph where each node denotes an AC and two nodes are linked if their community members are connected (line 9).
We repeat the iterations until the stopping condition $T$ (such as insufficient nodes or reaching the specified level limit) is met.
Since each iteration corresponds to one layer, all the ACs $HC$ can be organized into a multi-layer hierarchical tree structure, denoted by $\Delta$, where each community in one layer includes multiple communities in the next layer.
The extended version also provides more details.

\begin{algorithm}[ht]
  \small
  \caption{LLM-based hierarchical clustering}
  \label{alg:LHC}
   \SetKwInOut{Input}{input}\SetKwInOut{Output}{output}
    \Input{${G}(V,E)$, \texttt{Aug}, \texttt{GCluster},  $T$}
    $T\gets$ {\tt False}, ${HC}\gets \emptyset$\;
    \Repeat{$T$={\tt True}}{
        ${G}'(V,E') \gets$ \texttt{Aug}(${G}(V,E)$)\;
        
        \For{ each $e'=(u,v) \in E'$}{
            update the weight of $e'$ as $1 - \cos(z_u,z_v)$\;
        }

        $C \gets$ \texttt{GCluster}(${G}'(V,E')$)\;
        \lFor{each $c \in C$}{
            generate summary of $c$ by LLM
        }
        ${HC}\gets {HC} \cup C$\;
           
        ${G}(V,E) \gets$ build a new graph using $C$ and $E'$\; \tcp{\textcolor{teal}{update $T$ according to $G(V,E)$;}}
    }
    \Return{${HC}$;}
\end{algorithm}

\subsubsection{C-HNSW index.}
Given a query, to efficiently identify the most relevant information from each layer of the hierarchical tree $\Delta$, a naive method is to build a vector database for the ACs in each layer, which is costly in both time and space.
To tackle this issue, we propose to build a single hierarchical index for all the communities.
Recall that the ACs in $\Delta$ form a tree structure, and the number of nodes decreases as the layer level increases.
Since this tree structure is similar to the HNSW (Hierarchical Navigable Small World) index which is the most well-known index for efficient ANN search \shortcite{malkov2018efficient}, we propose to map entities and ACs of $\Delta$ into high-dimensional nodes, and then build a unified Community-based HNSW (C-HNSW) index for them.

$\bullet$ \textbf{The structure of C-HNSW.}
Conceptually, the C-HNSW index is a list of simple graphs with links between them, denoted by $\mathcal{H} = ({\mathcal G},{L_{inter}})$ with $\mathcal{\mathcal G}=\{G_0=(V_0, E_0), G_1=(V_1, E_1), \cdots, G_L=(V_L, E_L)\}$, where $G_i$ is a simple graph and each node of the simple graph corresponds to an attributed community or entity.
%
%
The number $L$ of layers in $\mathcal{H}$ is the same as that of $\Delta$.
Specifically, for each attributed community or entity in the $i$-th layer of $\Delta$, we map it to a high-dimensional node in the $i$-th layer of $\mathcal H$ by using a language model (e.g., nomic-embed-text \shortcite{nussbaum2024nomic}).

We next establish two types of links between these high-dimensional nodes, i.e., {\it intra-layer} and {\it inter-layer} links:
\begin{itemize}
    \item {\bf Intra-layer links:} These links exist between nodes in the same layers. Specifically, for each node in each layer, we link it to at least $M$ nearest neighbors within the same layer, where $M$ is a predefined value, and the nearest neighbors are determined according to a given distance metric $d$.
    Thus, all the intra-layer links are edges in all the simple graphs: $L_{intra}=\bigcup_{i=0}^{L}E_i$.

    \item {\bf Inter-layer links:} These links cross two adjacent layers. Specifically, we link each node in each layer to its nearest neighbor in the next layer. 
    As a result, all the inter-layer links can be represented as $L_{inter} = \bigcup_{i=1}^{L} \{(v,\psi(v))|v \in V_i,\psi(v)\in V_{i-1} \}$, 
        where $\psi(\cdot):V_i \rightarrow V_{i-1}$ is the injective function that identifies the nearest neighbor of each node in the lower layer.
\end{itemize}

For example, in Figure \ref{fig:overview}, the C-HNSW index has three layers (simple graphs), incorporating all the ACs.
Within each layer, each node is connected to its two nearest neighbors via intra-layer links, denoted by undirected edges.
The inter-layer links are represented by arrows, e.g., the green community at layer $L_1$ is connected to the green entity (its nearest neighbor at layer $L_0$).

Intuitively, since the two types of links above are established based on nearest neighbors, C-HNSW allows us to quickly search the relevant information for a query by traversing along with these links.
Note that C-HNSW is different from HNSW since it has intra-layer links and each node exists in only one layer.

$\bullet$ \textbf{The construction of C-HNSW.}
%
%
We propose a top-down approach to build a C-HNSW.
Specifically, by leveraging the query process of C-HNSW, which will be introduced in online retrieval, nodes are progressively inserted into the index starting from the top layer, connecting intra-layer links within the same layer and updating the inter-layer links.
For lack of space, we give the details of the construction algorithm in the extended version.

%% file: Section/Framework/online.tex
\subsection{Online retrieval}


%
In the online retrieval phase, after obtaining the query vector for a given question, ArchRAG generates the final answer by first conducting a hierarchical search on the C-HNSW index and then analyzing and filtering the retrieved information.

\subsubsection{Hierarchical search.} 
\label{sec:search}
We propose an efficient and fast retrieval algorithm, hierarchical search, to retrieve nodes from each layer of the C-HNSW structure.
Intuitively, retrieving nodes from a given layer in C-HNSW requires starting from the top layer and searching downward through two types of links (i.e.,  {\it intra-layer} and {\it inter-layer} links) to locate the nearest neighbors at the given layer.
In contrast, our hierarchical search algorithm accelerates this process by reusing intermediate results, the nearest neighbors found in higher layers, as the starting node for lower layers.
This approach avoids redundant computations that would otherwise arise from repeatedly searching from the top layer, thereby enabling efficient multi-layer retrieval.
%

Algorithm \ref{alg:query} illustrates hierarchical search.
Given the C-HNSW $\mathcal{H}$, query point $q$, and the number $k$ of nearest neighbors to retrieve at each layer, the hierarchical search algorithm can be implemented by the following iterative process:

\begin{enumerate}
    \item Start from a random node at the highest layer $L$, which serves as the starting node for layer $L$ (line 1).
    
    \item For each layer $i$ from the top layer $L$ down to layer $0$, the algorithm begins at the starting node and performs a greedy traversal (i.e., the {\tt SearchLayer} procedure) to find the set $R_i$ of the $k$ nearest neighbors of $q$. The set $R_i$ is then merged into the final result set $R$ (lines 4–5).
    
    \item The closest neighbor $c$ of $q$ is then obtained from $R_i$, and the algorithm proceeds to the next layer by traversing the inter-layer link of $c$, using it as the starting node for the subsequent search (lines 6–7).
\end{enumerate}

\begin{algorithm}[ht]
  \caption{{Hierarchical search}}
  \label{alg:query}
  \small
   \SetKwInOut{Input}{input}\SetKwInOut{Output}{output}
  \SetKwFunction{FSearchLayer}{SearchLayer}
  \SetKwProg{Fn}{Procedure}{:}{}
    \Input{$\mathcal{H} = ({\mathcal G},{L_{inter}})$, $q$, $k$.}
    $s\gets $ a random node in the highest layer $L$\;
    $R\gets \emptyset$\;
    \For{$i \gets L,\cdots,0$}{
    $R_i\gets$ \FSearchLayer(${ G_l}=(V_l,E_l), q, s, k$)\;
    $R\gets R\cup R_i$\;
    $c\gets$ get the nearest node from $R_i$\;
    $s\gets$ find the node in layer $i-1$ via the inter-layer links of $c$\;
    }
    \Return{$R$;}
    
     \Fn{\FSearchLayer{${G_i}=(V_i,E_i), q, s, k$}}{
    $V\gets \{s\}$,  $K\gets \{s\}$, $Q\gets $ initialize a queue containing $s$\; 
    \While{$|K|>0$}{
        $c \gets$ nearest node in $Q$\;
        $f \gets$ furthest node in $K$\;
        \lIf{$d(c,q)>d(f,q)$}{
            {\bf break}
        } 
        \For{each neighbor $x\in N(c)$ in $G_i$}{
            \lIf{$x \in V$}{\textbf{continue}}
            $V\gets V \cup \{x\}$\;
            $f \gets$ furthest node in $K$\;
            \If{$d(x,q)<d(f,q)$\text{ or }$|K|<k$ }{
                $Q\gets Q\cup \{x\}$, $K\gets K\cup \{x\}$\;
                \lIf{$|K|>k$}{remove $f$ from $K$}
            }
        }
    }
    \Return{${K}$;}
  }
\end{algorithm}

Specifically, the greedy traversal strategy compares the distance between the query point and the visited nodes during the search process.
It maintains a candidate expansion queue $Q$ and a dynamic nearest neighbor set $K$ containing $k$ elements, along with a stopping condition:

\begin{itemize}
    \item Expansion Queue $Q$: For each neighbor $x$ of a visited node, if $d(x,q)<d(f,q)$, where $f$ is the furthest node from $R$ to $q$, then $x$ is added to the expansion queue.
    \item Dynamic Nearest Neighbor Set $K$: Nodes added to $C$ are used to update $K$, ensuring that it maintains no more than $k$ elements, where $k$ is the number of query results.
    \item Stopping Condition: The traversal terminates if a node $x$ expanded from $Q$ satisfies $d(n,q) > d(n,f)$, where $f$ is the furthest node in $K$ from the query point $q$.
\end{itemize}

After completing the hierarchical search and obtaining the ACs and entities from each layer, we further extract their associated textual information.
In particular, at the bottom layer, we also extract the relationships between the retrieved entities, resulting in the textual subgraph representation denoted as $R_0$.
We denote all the retrieved textual information from each layer as $R_i$, where $i \in {0,1,\dots, L}$, which will be used in the adaptive filtering-based generation process.


%

\subsubsection{Adaptive filtering-based generation.}
While some optimized LLMs support longer text inputs, they may still encounter issues such as the ``lost in the middle'' dilemma \cite{liu2024lost}. 
Thus, direct utilization of retrieved information comprising multiple text segments for LLM-based answer generation risks compromising output accuracy.

To mitigate this limitation, we propose an adaptive filtering-based method that harnesses the LLM’s inherent reasoning capabilities.
We first prompt the LLM to extract and generate an analysis report from the retrieved information, identifying the parts that are most relevant to answering the query and assigning relevance scores to these reports.
Then, all analysis reports are integrated and sorted, ensuring that the most relevant content is used to summarize the final response to the query, with any content exceeding the text limit being truncated.
This process can be represented as:
\begin{align}
A_i &= LLM(P_{filter}||R_i) \\
Output&=LLM(P_{merge}||Sort(\{A_0,A_1,\cdots,A_n\}))
\end{align}
where $P_{filter}$ and $P_{merge}$ represent the prompts for extracting relevant information and summarizing, respectively, 
$A_i$, $i \in 0 \cdots n$ denotes the filtered analysis report.
The sort function orders the content based on the relevance scores from the analysis report.

%% file: Section/6_Experiment.tex
\section{Experiments}
\label{sec:exp}




\input{Section/Exp/setup}
\input{Section/Exp/main_results}
\input{Section/Exp/analysis}

%% file: Section/Exp/setup.tex
\subsection{Setup}

\begin{figure*}[ht]
\centering
\setlength{\belowcaptionskip}{-0.34cm}
\small
\renewcommand{\arraystretch}{1.1} 
\Large
\subfigure[Comprehensiveness]{
    \resizebox{0.23\textwidth}{!}{ 
    \begin{tabular}{c*{5}{c}}
  & VR & LR & C1 & C2 & AR \\ 
{VR} & \colorcell{50} & \colorcell{46} & \colorcell{18} & \colorcell{18} & \colorcell{1} \\ 
{LR} & \colorcell{54} & \colorcell{50} & \colorcell{21} & \colorcell{29} & \colorcell{16} \\ 
{C1} & \colorcell{82} & \colorcell{79} & \colorcell{50} & \colorcell{86} & \colorcell{18} \\ 
{C2} & \colorcell{82} & \colorcell{71} & \colorcell{14} & \colorcell{50} & \colorcell{16} \\ 
{AR} & \colorcell{99} & \colorcell{84} & \colorcell{82} & \colorcell{84} & \colorcell{50} \\ 
    \end{tabular}
    }
}
\subfigure[Diversity]{
    \resizebox{0.23\textwidth}{!}{ 
    \begin{tabular}{c*{5}{c}}
  & VR & LR & C1 & C2 & AR \\ 
{VR} & \colorcell{50} & \colorcell{64} & \colorcell{3} & \colorcell{12} & \colorcell{8} \\ 
{LR} & \colorcell{36} & \colorcell{50} & \colorcell{52} & \colorcell{63} & \colorcell{33} \\ 
{C1} & \colorcell{97} & \colorcell{48} & \colorcell{50} & \colorcell{52} & \colorcell{46} \\ 
{C2} & \colorcell{88} & \colorcell{37} & \colorcell{48} & \colorcell{50} & \colorcell{42} \\ 
{AR} & \colorcell{92} & \colorcell{67} & \colorcell{54} & \colorcell{58} & \colorcell{50} \\ 
    \end{tabular}
    }
}
\subfigure[Empowerment]{
    \resizebox{0.23\textwidth}{!}{ 
    \begin{tabular}{c*{5}{c}}
  & VR & LR & C1 & C2 & AR \\ 
{VR} & \colorcell{50} & \colorcell{39} & \colorcell{15} & \colorcell{21} & \colorcell{8} \\ 
{LR} & \colorcell{61} & \colorcell{50} & \colorcell{59} & \colorcell{30} & \colorcell{35} \\ 
{C1} & \colorcell{85} & \colorcell{41} & \colorcell{50} & \colorcell{60} & \colorcell{42} \\ 
{C2} & \colorcell{79} & \colorcell{70} & \colorcell{40} & \colorcell{50} & \colorcell{38} \\ 
{AR} & \colorcell{92} & \colorcell{65} & \colorcell{58} & \colorcell{62} & \colorcell{50} \\ 
    \end{tabular}
    }
}
\subfigure[Overall]{
    \resizebox{0.23\textwidth}{!}{ 
    \begin{tabular}{c*{5}{c}}
 & VR & LR & C1 & C2 & AR \\ 
{VR} & \colorcell{50} & \colorcell{46} & \colorcell{14} & \colorcell{18} & \colorcell{4} \\ 
{LR} & \colorcell{54} & \colorcell{50} & \colorcell{48} & \colorcell{31} & \colorcell{22} \\ 
{C1} & \colorcell{86} & \colorcell{52} & \colorcell{50} & \colorcell{70} & \colorcell{31} \\ 
{C2} & \colorcell{82} & \colorcell{69} & \colorcell{30} & \colorcell{50} & \colorcell{30} \\ 
{AR} & \colorcell{96} & \colorcell{78} & \colorcell{69} & \colorcell{70} & \colorcell{50} \\ 
    \end{tabular}
    }
}
\caption{Head-to-head win rates for abstract QA, comparing each row method against each column (higher is better).
VR, LR, and AR denote Vanilla RAG, HyLightRAG, and ArchRAG, respectively.}
\label{fig:winmap}
\end{figure*}

\begin{table}[tbp]
    \centering
    \small
    \begin{tabular}{c|rrr}
        \toprule
        Dataset & Multihop-RAG & HotpotQA & NarrativeQA  \\
        \midrule
        Passages & 609 & 9,221 & 1,572 \\
        Tokens & 1,426,396  & 1,284,956 & 121,152,448 \\
        Nodes & 23,353 & 37,436 & 650,571  \\
        Edges & 30,716 & 30,758 & 679,426 \\
        Questions & 2,556  & 1,000 & 43,304  \\
        Metrics & Acc, Rec & Acc, Rec & Blue, Met, Rou \\
        \bottomrule
    \end{tabular}
\caption{Datasets used in our experiments. Acc, Rec, Blue, Met, and Rou denote Accuracy, Recall, BLEU-1, METEOR, and ROUGE-L F1.
}
    \label{tab:dataset}
\end{table}

\paragraph{Datasets.}
We evaluate ArchRAG on both specific and abstract QA tasks.
For specific QA, we use Multihop-RAG~\shortcite{tang2024multihop}, HotpotQA \shortcite{yang2018hotpotqa}, and NarrativeQA \shortcite{kovcisky2018narrativeqa},
all of which are extensively utilized within the QA and Graph-based RAG research communities \shortcite{hu2022empowering,gutierrez2024hipporag,yao2022react,xu2024unsupervised,sarthi2024raptor,asai2023self}.
%
For abstract QA, we follow the GraphRAG \shortcite{edge2024local} method and reuse the Multihop-RAG corpus, prompting LLM to generate questions that convey a high-level understanding of dataset contents.
The statistics of these datasets are reported in Table \ref{tab:dataset}.

\paragraph{Baselines.}
Our experiments consider three configurations:
\begin{itemize}[left=0pt]
\item \textbf{Inference-only:} Using an LLM to answer questions without retrieval, i.e., {Zero-Shot} and {CoT} \shortcite{kojima2022large}.

\item \textbf{Retrieval-only:} Retrieval models extract relevant chunks from all documents and use them as prompts for LLMs. We select strong and widely used retrieval models: {BM25} \shortcite{robertson1994some} and {Vanilla RAG}.

\item \textbf{Graph-based RAG:} These methods leverage graph data during retrieval. We select RAPTOR~\shortcite{sarthi2024raptor}, {HippoRAG} \shortcite{gutierrez2024hipporag}, {GraphRAG}~\shortcite{edge2024local}, and {LightRAG}~\shortcite{guo2024lightrag}.
Particularly, {GraphRAG} has two versions, i.e., {GGraphRAG} and {LGraphRAG}, which use global and local search methods, respectively.
Similarly, {LightRAG} integrates local search, global search, and hybrid search, denoted by {LLightRAG}, {HLightRAG}, and {HyLightRAG}, respectively.
\end{itemize}

In {GGraphRAG}, all communities below the selected level are first retrieved, and then the LLM is used to filter out irrelevant communities.
This process can be viewed as utilizing the LLM as a {\it retriever} to find relevant communities within the corpus.
According to the selected level of communities~\shortcite{edge2024local}, {GGraphRAG} can be further categorized into {C1} and {C2}, representing high-level and intermediate-level communities, respectively, with {C2} as the default.
%

\paragraph{Metrics \& Implementation.}
For the specific QA tasks, we use Accuracy and Recall to evaluate performance on the first two datasets based on whether gold answers are included in the generations instead of strictly requiring exact matching, following \shortcite{schick2024toolformer,mallen2022not,asai2023self}.
We also use the official metrics of BLEU, METEOR, and ROUGE-L F1 in the NarrativeQA dataset.
For the abstract QA task, we follow prior work \shortcite{edge2024local} and adopt a head-to-head comparison approach using an LLM evaluator (GPT-4o).
Overall, we utilize four evaluation dimensions: Comprehensiveness, Diversity, Empowerment, and Overall. 
For implementation, we mainly use Llama 3.1-8B \shortcite{dubey2024llama} as the default LLM and use nomic-embed-text \shortcite{nussbaum2024nomic} as the text embedding model.
We use KNN for graph augmentation and the weighted Leiden algorithm for community detection.
For retrieval item $k$, we search the same number of items at each layer, with $k=5$ as the default.
All methods are required to complete index construction and query execution within 3 days, respectively.

%% file: Section/Exp/main_results.tex
\definecolor{c1}{RGB}{216,174,174} 
\definecolor{c2}{RGB}{224,175,107}  
\definecolor{c3}{RGB}{222,117,123} 
\definecolor{c4}{RGB}{250, 221, 104} 

\definecolor{c5}{RGB}{174,223,172} 

\definecolor{c9}{RGB}{190,150,210} 

\definecolor{c6}{RGB}{138,170,214}  
\definecolor{c11}{RGB}{172, 196, 226} 
\definecolor{c12}{RGB}{208, 221, 238} 

\definecolor{c7}{RGB}{248,199,1} 
\definecolor{c21}{RGB}{163,137,214} 

\definecolor{c8}{RGB}{255,0,127} 



\definecolor{c10}{RGB}{230,189,69} %

\definecolor{c13}{RGB}{115,107,157} %
\definecolor{c14}{RGB}{208,108,157} %

\pgfplotstableread[row sep=\\,col sep=&]{
datasets & CSH & CSSH & SSCS & SSCSMP & Ecore & Vcore & SCANS \\ 
2 & 3.132 & 3.5 & 3.284 & 2.267 & 2.45 & 2.075 & 7.317 \\
4 & 8.327 & 9.288 & 3.143 & 2.363 & 3.381 & 3.122 & 4.066 \\
6 & 12.593 & 6.744 & 1.581 & 1.619 & 10.857 & 10.833 & 5.621 \\
}\DIACom
\pgfplotstableread[row sep=\\,col sep=&]{
datasets & CSH & CSSH & SSCS & SSCSMP & Ecore & Vcore & SCANS \\ 
2 & 0.57 & 0.629 & 0.889 & 0.898 & 0.85 & 0.64 & 0.266 \\
4 & 0.137 & 0.153 & 0.778 & 0.857 & 0.335 & 0.36 & 0.211 \\
6 & 0.536 & 0.708 & 0.822 & 0.825 & 0.425 & 0.594 & 0.568 \\
}\CCcom

\pgfplotstableread[row sep=\\,col sep=&]{
datasets & zero & cot & BM25 & vanilla & raptor & hippo & lightl & lighth & lighthy & graphl & graphg & ArchRAG \\ 
2 & 436.93 & 988.13 & 2590.56 & 3042.43 & 3400.56 & 2359.64 & 5193.792 & 7152.56 & 7120.468 & 22118.98 & 246373.92 & 8890.07 \\ 
4 & 211.26 & 488.74 & 351.49 & 332.37 & 0 & 1392.55 & 1673.48 & 1852.6 & 2843.442 & 11457.59 & 123030.84 & 2024.94 \\ 
}\Qtime
\pgfplotstableread[row sep=\\,col sep=&]{
datasets & zero & cot & BM25 & vanilla & raptor & hippo & lightl & lighth & lighthy & graphl & graphg & ArchRAG \\ 
2 & 426105 & 795941 & 5594336 & 5556774 & 8468794 & 6121633 & 14617559 & 15467710 & 18515561 & 1536156 & 2726371287 & 12575315 \\ 
4 & 89649 & 172217 & 675569 & 203654 & 0 & 886793 & 3088800 & 3223300 & 4817210 & 601000 & 1394436159 & 5191529 \\ 
}\Qtoken

\pgfplotstableread[row sep=\\,col sep=&]{
datasets & raptor & Hippo & light & GraphRAG & ArchRAG \\ 
2 & 4200 & 3370 & 6689.42 & 7591 & 9320.95 \\
4 & 100000 & 4521 & 18963.89 & 16118 & 14383.94 \\
}\Indextime
\pgfplotstableread[row sep=\\,col sep=&]{
datasets & raptor  & Hippo & light & GraphRAG & ArchRAG \\ 
2 & 4542276 & 3044748 & 12537458 & 12212572 & 13328928 \\
4 & 1000000 & 8023838 & 52436836 & 29140517 & 35248693 \\
}\Indextoken

\begin{table*}[ht]
    \centering
    \small
    \begin{tabular}{llccccccc}
        \toprule
        \multirow{2}{*}{Baseline Type} & \multirow{2}{*}{Method} & \multicolumn{2}{c}{Multihop-RAG} & \multicolumn{2}{c}{HotpotQA} & \multicolumn{3}{c}{NarrativeQA} \\
        \cmidrule(lr){3-4} \cmidrule(lr){5-6} \cmidrule(lr){7-9}
        & &(Accuracy) & (Recall) & (Accuracy) & (Recall) & (BLEU-1) & (METEOR) & (ROUGE-L F1)\\
\midrule
\multirow{2}{*}{Inference-only} & {Zero-shot} & 47.7 & 23.6 & 28.0 & 31.8 & \underline{8.0} & 7.9 & 8.6 \\
& {CoT} & 54.5 & 28.7 & 32.5 & 39.7 & 5.0 & 8.1 & 6.4 \\
\midrule
\multirow{2}{*}{Retrieval-only} 
& {BM25} & 37.6 & 19.4 & 49.7 & 53.6 & 2.0 & 4.9  & 2.8 \\
& {Vanilla RAG} & 58.6 & 31.4 & 50.6 & 56.1 & 2.0 & 4.9 & 2.8 \\
\midrule
\multirow{6}{*}{Graph-based RAG} 
& {RAPTOR} & \underline{59.1} & \underline{34.1} & N/A & N/A & 5.5 & \underline{12.5} & \underline{9.1} \\
& {HippoRAG} & 38.9 & 19.1 & \underline{51.3} & \underline{56.8} & 2.2 & 5.0 & 2.8 \\
& {LLightRAG} & 44.1 & 25.1 & 34.1 & 41.8 & 4.5 & 8.7 & 6.6 \\
& {HLightRAG} & 48.5 & 28.7 & 25.6 & 33.3 & 4.4 & 8.1 & 6.1 \\
& {HyLightRAG} & 50.3 & 30.3 & 35.6 & 43.3 & 5.0 & 9.4 & 7.0 \\
& {LGraphRAG} & 40.1 & 23.8 & 29.7 & 35.5 & 3.9 & 3.3 & 3.5 \\
& {GGraphRAG}& 45.9 & 28.4 & 33.5 & 42.6 & OOT & OOT & OOT \\
\midrule
Our proposed & {ArchRAG} & {\bf 68.8} & {\bf 37.2 }& {\bf 65.4} & {\bf 69.2} & {\bf 11.5} & {\bf 15.6} & {\bf 17.6} \\
        \bottomrule
    \end{tabular}
    \caption{Performance comparison of different methods across various datasets for solving specific QA tasks. The best and second-best results are marked in bold and underlined. 
    %
    OOT: Not finished within 3 days.}
    \label{tab:docqa}
\end{table*}

\subsection{Overall results}
\label{sec:overallExp}

We compare our method with baseline methods in solving both abstract and specific QA tasks.

$\bullet$ \textbf{Results of abstract QA tasks.}
We compare {ArchRAG} against baselines across four dimensions on the Multihop-RAG dataset.
For the LightRAG, we only compare the {HyLightRAG} method, as it represents the best version~\shortcite{guo2024lightrag}.
As shown in Figure \ref{fig:winmap}, GGraphRAG outperforms other baseline methods, while our method achieves comparable performance on the diversity and empowerment dimensions and significantly surpasses it on the comprehensive dimension.
Overall, by leveraging ACs, ArchRAG demonstrates superior performance in addressing abstract QA tasks.

$\bullet$ \textbf{Results of specific QA tasks.}
Table \ref{tab:docqa} reports the performance of each method on three datasets.
Note that GGraphRAG fails to complete querying on the NarrativeQA dataset within the 3-day time limit.
RAPTOR is unable to build the index on datasets like HotpotQA, which contains a large number of text chunks. 
Its Gaussian Mixture Model (GMM) clustering algorithm requires prohibitive computational time and suffers from non-termination issues during clustering.
Clearly, {ArchRAG} demonstrates a substantial performance advantage over other baseline methods on these datasets.
The experimental results suggest that not all communities are suitable for specific QA tasks, as the {GGraphRAG} performs poorly. 
Furthermore, GraphRAG does not consider node attributes during clustering, which causes the community's summary to become dispersed, making it difficult for the LLM to extract relevant information from a large number of communities.
Thus, we gain an interesting insight: {\it LLM may not be a good retriever, but is a good analyzer.}
%
We further analyze the reasons behind the underperformance of each graph-based RAG method and support our claims with empirical evidence in the extended version.

\begin{figure}[tb]
    \centering
    \setlength{\belowcaptionskip}{-0.2cm}
    \quad \ref{eff_leg}\\
    \subfigure[Time cost]{
    \begin{tikzpicture}[scale=0.45]
            \begin{axis}[
                ybar=0.5pt,
                bar width=0.6cm,
                width=\textwidth,
                height=0.26\textwidth,
                xtick=data,
                xticklabels={\huge Multihop-RAG,\huge HotpotQA},
                legend style={
                 at={(0.0,0.78)}, 
    anchor=south west, 
                    align=left, 
                legend columns=4,
                draw=none},
                legend image code/.code={
                    \draw [#1, line width=0.5pt] (0cm,-0.1cm) rectangle (0.2cm,0.2cm); },
                legend to name=eff_leg,
                xmin=1,xmax=5,
                ymin=100,ymax=1000000,
                ytick = {100, 1000, 10000, 100000},
                yticklabels = {$10^{2}$, $10^{3}$, $10^{4}$, $10^{5}$},
                ymode = log,
                tick align=inside,
                ticklabel style={font=\Huge},
                every axis plot/.append style={line width = 2.5pt},
                every axis/.append style={line width = 2.5pt},
                ylabel={\textbf{\huge time (s)}}
                ]
\addplot[fill=c1] table[x=datasets,y=zero]{\Qtime};
\addplot[fill=c2] table[x=datasets,y=cot]{\Qtime};
\addplot[fill=c3] table[x=datasets,y=BM25]{\Qtime};
\addplot[fill=c4] table[x=datasets,y=vanilla]{\Qtime};
\draw[c9] (axis cs:3.77,100) node[rotate=90, anchor=west] {\Huge \bf N/A};
\addplot[fill=c9] table[x=datasets,y=raptor]{\Qtime};
\addplot[fill=c5] table[x=datasets,y=hippo]{\Qtime};
\addplot[fill=c6] table[x=datasets,y=lightl]{\Qtime};
\addplot[fill=c11] table[x=datasets,y=lighth]{\Qtime};
\addplot[fill=c12] table[x=datasets,y=lighthy]{\Qtime};
\addplot[fill=c7] table[x=datasets,y=graphl]{\Qtime};
\addplot[fill=c21] table[x=datasets,y=graphg]{\Qtime};
\addplot[fill=c8] table[x=datasets,y=ArchRAG]{\Qtime};

\legend{\small {Zero-Shot},\small {CoT},\small {BM25}, \small {Vanilla RAG}, \small {RAPTOR}, \small {HippoRAG}, \small {LLightRAG}, \small {HLightRAG}, \small {HyLightRAG}, \small {LGraphRAG}, \small {GGraphRAG}, \small {ArchRAG}}
            \end{axis}
        \end{tikzpicture}
        \label{fig:query-time}
    }
    \subfigure[Token cost]{
		\begin{tikzpicture}[scale=0.45]
            \begin{axis}[
                ybar=0.5pt,
                bar width=0.6cm,
                width=\textwidth,
                height=0.26\textwidth,
                xtick=data,
                xticklabels={\huge Multihop-RAG,\huge HotpotQA},
                xmin=1,xmax=5,
                ymin=50000,ymax=7000000000,
                ytick = {1000000, 10000000, 100000000, 1000000000},
                yticklabels = {$10^{0}$, $10^{1}$, $10^{2}$, $10^{3}$},
                ymode = log,
                tick align=inside,
                ticklabel style={font=\Huge},
                every axis plot/.append style={line width = 2.5pt},
                every axis/.append style={line width = 2.5pt},
                ylabel={\textbf{\huge token (M)}}
                ]

\addplot[fill=c1] table[x=datasets,y=zero]{\Qtoken};
\addplot[fill=c2] table[x=datasets,y=cot]{\Qtoken};
\addplot[fill=c3] table[x=datasets,y=BM25]{\Qtoken};
\addplot[fill=c4] table[x=datasets,y=vanilla]{\Qtoken};
\addplot[fill=c9] table[x=datasets,y=raptor]{\Qtoken};
\draw[c9] (axis cs:3.77,50000) node[rotate=90, anchor=west] {\Huge \bf N/A};
\addplot[fill=c5] table[x=datasets,y=hippo]{\Qtoken};
\addplot[fill=c6] table[x=datasets,y=lightl]{\Qtoken};
\addplot[fill=c11] table[x=datasets,y=lighth]{\Qtoken};
\addplot[fill=c12] table[x=datasets,y=lighthy]{\Qtoken};
\addplot[fill=c7] table[x=datasets,y=graphl]{\Qtoken};
\addplot[fill=c21] table[x=datasets,y=graphg]{\Qtoken};
\addplot[fill=c8] table[x=datasets,y=ArchRAG]{\Qtoken};
            \end{axis}
        \end{tikzpicture}
        \label{fig:query-token}
	}
    \caption{Comparison of query efficiency.}
    \label{fig:query-efficiency}
\end{figure}

$\bullet$ \textbf{Efficiency of {ArchRAG}.}
We compare the time cost and token usage of {ArchRAG} with those of other baseline methods.
As shown in Figure \ref{fig:query-efficiency}, {ArchRAG} demonstrates significant time and cost efficiency for online queries.
For example, token usage on the HotpotQA dataset is cut by $\bf 250\times$ with ArchRAG compared to GraphRAG-Global, from 1,394M tokens down to 5.1M tokens.

To further evaluate ArchRAG, we test the {\bf efficiency of hierarchical search}, {\bf indexing performance}, and {\bf effectiveness on an additional dataset}, RAG-QA Arena ~\shortcite{han2024rag}.
Results show that ArchRAG achieves up to $5.4\times$ faster retrieval than basic HNSW, maintains efficient indexing, and achieves state-of-the-art performance on the RAG-QA Arena dataset.
Additional details are provided in the extended version.


%% file: Section/Exp/analysis.tex
\subsection{Detailed Analysis}



To better understand the effectiveness of our proposed ArchRAG, we perform the following ablation study and experiment with a GGraphRAG variant.

\begin{table}[tb]
    \centering
    \small
    \begin{tabular}{lcccc}
        \toprule
\multirow{2}{*}{Method variants} &  \multicolumn{2}{c}{Multihop-RAG} & \multicolumn{2}{c}{HotpotQA} \\
\cmidrule(lr){2-3} \cmidrule(lr){4-5}
        & (Acc) & (Rec) & (Acc) & (Rec)\\
\midrule
{ ArchRAG}  & 68.8 & 37.2 & 65.4 & 69.2 \\
\footnotesize{ - Spec.}  & 67.1 & 36.7 & 60.5 & 63.7 \\
\footnotesize{ - Spec. (No Aug)}  & 62.8 & 34.2 & 64.8 & 63.2 \\
\footnotesize{ - Leiden}  & 63.2 & 34.1 & 61.7 & 64.8 \\
\footnotesize{ - Single-Layer}  & 63.8 & 36.4 & 60.1 & 63.6 \\
\footnotesize{ - Entity-Only}  & 61.2 & 34.7 & 59.9 & 63.1 \\
\footnotesize{ - Direct Prompt}  & 59.9 & 29.6 & 40.7 & 45.4 \\
        \bottomrule
    \end{tabular}
\caption{Comparing the performance of different variants of ArchRAG on the specific QA tasks. Acc and Rec denote Accuracy and Recall, respectively.}
    \label{tab:variant}
\end{table}

$\bullet$ \textbf{Ablation study.}
%
To evaluate the contributions of different components, we design several ArchRAG variants and conduct ablation experiments.
These include two modifications to the LLM-based hierarchical clustering framework, three targeting core design elements in ArchRAG—attributes, hierarchy, and communities—and one direct prompting variant, as detailed below:
\begin{itemize}
    \item {Spec.}: Spectral clustering instead of weighted Leiden.
    \item {Spec. (No Aug)}: Spectral clustering without graph augmentation.
    \item {Leiden}: Replaces our clustering framework with Leiden.
    \item {Single-Layer}: Replaces our hierarchical index and search with a single-layer community.
    \item {Entity-Only}: Generate the response using entities only.
    \item {Direct Prompt}: Direct prompts the LLM to generate the response without the adaptive filtering-based generation.
\end{itemize}

%
As shown in Table \ref{tab:variant}, the performance of {ArchRAG} on specific QA tasks decreases when each feature is removed, with the removal of the community component resulting in the most significant drop.
Additionally, the direct variant demonstrates that the adaptive filtering-based generation process can effectively extract relevant information from retrieved elements.

$\bullet$ \textbf{Impact of Attributed Communities in RAG. } We propose a new variant, \textit{GraphRAG+AC}, which replaces the original Leiden-based communities in GraphRAG with our ACs, while preserving the original Global Search pipeline.
As shown in Table~\ref{tab:varison-gr}, this variant results in a significant performance improvement compared to the original approach.
Specifically, on the HotpotQA dataset, GraphRAG+AC improves accuracy by 51\% compared to GGraphRAG.

\begin{table}[tb]
    \centering
    \small
    \begin{tabular}{c|cccc}
\toprule
\multirow{2}{*}{Method} &  \multicolumn{2}{c}{Multihop-RAG} & \multicolumn{2}{c}{HotpotQA} \\
\cmidrule(lr){2-3} \cmidrule(lr){4-5}
&  (Acc) & (Rec) & (Acc) & (Rec)\\
\midrule
GGraphRAG & 45.9 & 28.4 & 33.5 & 42.6 \\
GraphRAG+AC & 49.3 & 31.4 & 50.6 (51.0\% $\uparrow$) & 52.8 \\
ArchRAG & 68.8 & 37.2 & 65.4 & 69.2 \\
\bottomrule
    \end{tabular}
    \caption{Results of GraphRAG variants using our ACs.}
    \label{tab:varison-gr}
\end{table}

We further test ArchRAG under {\bf different LLM backbones}, various {\bf top-$k$ retrieval} settings, evaluate the {\bf community quality} of its LLM-based hierarchical clustering, and also provide additional \textbf{case studies} of our ArchRAG.
%
Detailed experiments are provided in the extended version.



%% file: Section/8_Conclusion.tex
\section{Conclusion}
\label{sec:conclusion}

In this paper, we propose ArchRAG, a novel graph-based RAG approach, by augmenting the question using attributed communities from the knowledge graph built on the external corpus, and building a novel index for efficient retrieval of relevant information.
Our experiments show that ArchRAG is highly effective and efficient.
In the future, we will explore fast parallel graph-based RAG methods to process large-scale external corpus.


%% file: Section/9_Appendix.tex
\appendix
\clearpage

\section{Method details of ArchRAG}

\subsection{LLM-based hierarchical clustering}
\label{sec:LHC}

For example, as shown in the second step of offline indexing in Figure \ref{fig:overview}, LLM-based hierarchical clustering first enhances the original KG at layer $L_0$ by adding similar edges.
Next, the weight of each edge is calculated based on the strength of the relationship between the node embeddings, and a weighted clustering algorithm is applied to obtain four communities.
Based on the existence of links between nodes within a community, the topology of communities at layer $L_1$ is constructed, resulting in a graph of communities. 
This process is repeated, ultimately generating a clustering result consisting of three layers of hierarchical communities.

        

           

\subsection{More details of C-HNSW}
\label{sec:c_hnsw}


$\bullet$ \textbf{Introduction of HNSW. }
We provide a brief introduction to HNSW, an efficient Approximate Nearest Neighbor Search (ANNS) technique for vector databases.

\begin{definition}[Hierarchical Navigable Small World (HNSW) \cite{malkov2018efficient}]
    HNSW is a graph-based ANNS algorithm that consists of a multi-layered index structure, where each node uniquely corresponds to a vector in the database.
    Given a set $S$ containing $n$ vectors, the constructed HNSW can be represented as a pair $\mathcal{H} = (\mathcal{\bf G}, \mathcal{C})$. 
    $\mathcal{\bf G}=\{G_0, G_1, \dots, G_L\}$ is a set of simple graphs (also called layers) $G_i=(V_i, E_i)$, where $i \in \{0,1,\dots, L\}$ and $V_L \subset V_{L-1} \subset \dots \subset V_1 \subset V_0 = S$.
    $\mathcal{C}$ records the inter-layer mappings of edges between the same node across adjacent layers $\mathcal{C} = \bigcup_{i=0}^{L-1} \{(v,\phi(v))|v \in V_i,\phi(v)\in V_{i+1} \}$, where $\phi(v):V_i \rightarrow V_{i+1}$ is the mapping function for the same node across two adjacent layers.
\end{definition}

The nodes in the multi-layer graph of HNSW are organized in a nested structure, where each node at each layer is connected to its nearest neighbors. 
During a query, the search begins at the top layer and quickly identifies the node closest to $q$ through a greedy search.
Then, through inter-layer mapping, the search proceeds to the next lower layer.
This process continues until all approximate nearest neighbors are identified in $G_0$.

$\bullet$ \textbf{The construction of C-HNSW. } 
A naive approach to build the C-HNSW index is to build nodes first and then establish the two types of links by finding the nearest neighbors of each node.
However, the process of finding the nearest neighbor is costly.
To accelerate the construction, we propose a top-down approach by borrowing the idea of HNSW construction.
The construction of C-HNSW is illustrated in Algorithm \ref{alg:construction}.
The construction algorithm of C-HNSW follows a top-down approach. 
Using the query process of C-HNSW, we obtain the $M$ nearest neighbors of each node in its layer, and the inter-layer links are continuously updated during this process.
When inserting a node $x$ at layer $i$ , if the nearest neighbor at the layer $i+1$ is $c_j$, the inter-layer link of $c_j$ is updated in the following two cases:
\begin{itemize}
    \item Node $c_j$ does not have a inter-layer link to the layer $i$.
    \item The distance from node $c_j$ to node $x$ is smaller than the distance from $c_j$ to its previous nearest neighbor $x'$, i.e., $d(c_j,x) < d(c_j,x')$.
\end{itemize}

After all nodes at layer $i$ ($i<L$) have been inserted, we check each node at layer $i+1$ to ensure it has a inter-layer link, confirming the traverse from the higher layer to the lower layer.

\begin{algorithm}[ht]
  \caption{{C-HNSW construction}}
  \label{alg:construction}
  \small
   \SetKwInOut{Input}{input}\SetKwInOut{Output}{output}
    \Input{The Hierarchical community $\mathcal{HC}$, KG $\mathcal{G}(V,E)$, maximum number of connections for each node $M$.}
    $\mathcal{H}\gets\emptyset$\;
    
    ${\bf V}\gets$ \{$\mathcal{HC} \cup V$\} \tcp{\textcolor{teal}{Get all nodes of each layer.}}
    \For{each layer $l \gets L \cdots 0$}{
        \For{each node $v \in V_l$}{
            \If{$l\neq L$}{
                $R\gets$ SearchLayer(${ G_l}=(V_l,E_l), q, s, 1$)\;
                $c\gets$ get the nearest node from R\;
                $s\gets$ node in layer $l-1$ via $c$'s inter-layer link\;
                \If{$s$ is null\texttt{ or }$d(c,s)>d(c,v)$}{
                update $v$ as $c$'s inter-layer link
                }
            }
            \lIf{$l=L$\texttt{ or }$s$ is null}{$s\gets $ random node in layer $l$}
            $R\gets$ SearchLayer(${ G_l}=(V_l,E_l), q, s, M$)\;
            add edges between $v$ and $R$, update $E_l$\; 
        }
        $\mathcal{H}\gets \mathcal{H}\cup { G_l}=(V_l, E_l)$
    }
    \Return{$\mathcal{H}$;}
\end{algorithm}

\subsection{Complexity analysis of ArchRAG}
\label{sec:complexity}

We now analyze the complexity of our ArchRAG approach.
Since the token cost is more important in the era of LLM, we replace the space complexity with the token cost.

The offline indexing process of ArchRAG includes the KG construction, hierarchical clustering, and C-HNSW construction.
The time complexity and token usage are as follows:


\begin{lemma}
\label{lemma:index-time}
Given a large text corpus or a large set of text documents with a total of $D$ tokens, the time complexity of the offline indexing process of ArchRAG is $O(I\frac{D}{w} + \frac{1-a^L}{1-a}(n*t+I\frac{D}{w}+\pi(m)+n\log n))$, where ${\bf I}$ is the generation time of the LLM for a single inference, $w$ is the specified token size of one chunk, $n$ and $m$ are the number of entities and relations in the extracted KG, $L$ is the height of the resulting hierarchical community structure, and $a$ is the average ratio of the number of nodes between two consecutive layers, with $0<a<1$. For a given embedding model, the computation time for the embedding of an entity description is denoted by $t$, while a specific clustering method is typically a function of $m$, represented as $\pi(m)$.
\end{lemma}

\begin{proof}
    For the corpus $D$, we use the LLM to infer and extract the KG from each chunk of size $w$, resulting in a cost of $O(I\frac{D}{w})$ for constructing the KG. Since the size of all community summaries in a single layer would not exceed the length of the corpus, their LLM inference time is also less than $O(I\frac{D}{w})$. For each layer of clustering, the embedding of each point must be computed, and clustering is performed with time complexity of $\pi(m)$. In the C-HNSW construction, each point requires $O(nlogn)$ time to perform the k-nearest neighbor search and establish connections, which is similar to the proof in \cite{malkov2018efficient}. For an L-layer multi-layer graph structure, where the number of nodes decreases by a factor of $a$ between two consecutive layers, the increase in clustering and C-HNSW construction time is given by: $\frac{1 - a^L}{1 - a}$.
\end{proof}

\begin{lemma}
\label{lemma:index-cost}
Given a large text corpus or a large set of text documents with a total of $D$ tokens, the number of tokens used in the offline indexing process of ArchRAG is $O(D(1+\frac{1-a^L}{1-a}))$, where $L$ is the height of the resulting hierarchical community structure and $a$ is the average ratio of the number of nodes between two consecutive layers, with $0<a<1$.
\end{lemma}

\begin{proof}
    Based on the above proof, the token cost for constructing the KG is $O(D)$, while the token cost for the summaries does not exceed $O(D\frac{1-a^L}{1-a})$.
\end{proof}

Generally, $L$ is $O(\log n)$, where $n$ is the number of extracted entities. 
However, due to the constraints of community clustering, $L$ is typically constant, usually no greater than 5.

Next, we analyze the time complexity and token usage in the online query process of the ArchRAG.

\begin{lemma}
\label{lemma:query-time}
Given a C-HNSW with $L$ layers constructed from a large text corpus or a large set of text documents, the time complexity of a sequentially executed single online retrieval query in ArchRAG is $O(e+ LkI+Lk\log(n))$, where $I$ is the generation time of the LLM for a single inference, $e$ is the time cost of computing the query embedding, $k$ is the number of nodes retrieved at each layer, and $n$ is the number of nodes at the lowest layer in C-HNSW.
\end{lemma}

\begin{proof}
    ArchRAG first computes the embedding of the query, which takes $O(e)$ time. For each layer, querying one nearest neighbor takes no more than $O(\log(n))$ time, similar to the proof in \cite{malkov2018efficient}. In the Adaptive filtering-based generation, the content of each query is analyzed and inferred, requiring $O(LkI)$ time. Therefore, the total time for the online retrieval query is $O(e+ LkI+Lk\log(n))$.
\end{proof}

\begin{lemma}
\label{lemma:query-cost}
Given a C-HNSW with $L$ layers constructed from a large text corpus or a large set of text documents, the number of tokens used for a single online retrieval query in ArchRAG is $O(kL(c+P)))$, where $k$ is the number of neighbors retrieved at each layer, $c$ is the average token of the retrieved content, and $P$ is the token of the prompt.
\end{lemma}

\begin{proof}
    The token consumption for analyzing all retrieved information is $O(kL(c+P)))$, while the token consumption for generating the final response is of constant order. Therefore, the total token consumption for the online retrieval is $O(kL(c+P)))$.
\end{proof}

In practice, multiple retrieved contents are combined, allowing the LLM to analyze them together, provided the total token count does not exceed the token size limit. 
As a result, the time and token usage for online retrieval are lower than those required for analysis.

\section{Experimental details}

\subsection{Metrics}
\label{sec:metrics}

This section provides additional details on the metrics.

$\bullet$ \textbf{Metrics for specific QA tasks. } We choose accuracy as the evaluation metric based on whether the gold answers are included in the model's generations rather than strictly requiring an exact match, following \cite{schick2024toolformer,mallen2022not,asai2023self}.
This is because LLM outputs are typically uncontrollable, making it difficult for them to match the exact wording of standard answers.
Similarly, we choose recall as the metric instead of precision, as it better reflects the accuracy of the generated responses.
Additionally, when calculating recall, we adopt the same approach as previous methods \cite{gutierrez2024hipporag,asai2023self}: if the golden answer or the generated output contains ``yes'' or ``no'', the recall for that question is set to 0. 
Therefore, the recall metric is not perfectly correlated with accuracy.

$\bullet$ \textbf{Metrics for abstract QA tasks. }
Following existing works, we use an LLM to generate abstract questions, with the prompts shown in Figure \ref{fig:prompt_summary}, defining ground truth for abstract questions, particularly those involving complex high-level semantics, poses significant challenges.
We build on existing works \cite{edge2024local,guo2024lightrag} to address this and adopt an LLM-based multi-dimensional comparison method (including comprehensiveness, diversity, empowerment, and overall).
We employ a robust LLM, specifically GPT-4o, to rank each baseline against our method.
Figure \ref{fig:eval_summary} shows the evaluation prompt we use.

$\bullet$ \textbf{Metrics of community quality. } We select the following metrics to evaluate the quality of the community:
\begin{enumerate}
    \item \textit{Calinski-Harabasz Index (CHI)} \cite{calinski1974CHIndex}: A higher value of CHI indicates better clustering results because it means that the data points are more spread out between clusters than they are within clusters. It is an internal evaluation metric where the assessment of the clustering quality is based solely on the dataset and the clustering results and not on external ground-truth labels. The CHI is calculated by between-cluster separation and within-cluster dispersion:
    \begin{equation}
    CHI = \frac{N-C}{C-1}\frac{\sum^{C}_{i=1}n_i||c_i-c||^2}{\sum^{C}_{i=1}\sum_{x \in C_i}||x-c_i||^2}.
    \end{equation}
    $N$ is the number of nodes. $C$ is the number of clusters.$n_i$ is the number of nodes in cluster $i$. $C_i$ is the $i-th$ cluster. $c_i$ is the centroid of cluster $C_i$. $c$ is the overall centroid of the datasets. $x$ is the feature of the target node.
    
    \item \textit{Cosine Similarity (Sim)} \cite{charikar2002similarity}: Cosine similarity is a measure of similarity between two non-zero vectors defined in an inner product space. In this paper, for each cluster, we calculate the similarity between the centroid of this cluster and each node in this cluster:
    \begin{equation}
    Sim = \frac{1}{N}\sum^{C}_{i=1}\sum_{x \in C_i}Cosine(x, c_i).
    \end{equation}
    $N$ is the number of nodes. $C$ is the number of clusters. $C_i$ is the $i$-th cluster. $c_i$ is the centroid of cluster $C_i$. $x$ is the feature of the target node.
    \begin{equation}
    Cosine(x, y) = \frac{\sum_{i=1}^{n}x_i y_i}{\sqrt{\sum_{i=1}^{n}x_i^2}\sqrt{\sum_{i=1}^{n}y_i^2}}.
    \end{equation}
\end{enumerate}


\subsection{Implementation details}
\label{sec:imp}

We implement our ArchRAG in Python, while C-HNSW is implemented in C++ and provides a Python interface for integration.
We implement C-HNSW using the FAISS framework and employ the inner product metric to measure the proximity between two vectors. 
All the experiments were conducted on a Linux operating system running on a machine with an Intel Xeon 2.0 GHz CPU, 1024GB of memory, and 8 NVIDIA GeForce RTX A5000 GPUs, each with 24 GB of memory.
All methods utilize 10 concurrent LLM calls, and to maintain consistency, other parallel computations in the method, such as embedding calculations, also use 10 concurrent threads. 
%
%
Figure \ref{fig:prompt_afg} demonstrates the prompt used in Adaptive filtering-based generation.
Please refer to our repository (https://anonymous.4open.science/r/H-CAR-AG-1E0B/) to view the detailed prompts.

$\bullet$ \textbf{Details of clustering methods. }
The graph augmentation methods we choose are the KNN algorithm, which computes the similarity between each node, and CODICIL~\cite{CODICIL2013efficient}, which selects and adds the top similar edges to generate better clustering results. 
In the KNN method, we set the $K$ value as the average degree of nodes in the KG.

\subsection{Additional experiments}
\label{sec:more-exp}

$\bullet$ \textbf{Efficiency of hierarchical search (C-HNSW).}
To evaluate the efficiency of C-HNSW, we conduct experiments on a synthetic hierarchical dataset comprising 11 layers.
The bottom layer (Layer 0) contains 10 million nodes, and the number of nodes decreases progressively across higher layers by randomly dividing each layer’s size by 3 or 4.
The top layer (Layer 10), for example, contains only 74 nodes.
This hierarchical structure simulates the process of LLM-based hierarchical clustering.
Each node is assigned a randomly generated 3072-dimensional vector, simulating high-dimensional embeddings such as those produced by text or image encoders (e.g., \texttt{text-embedding-3-large}, used in ChatGPT, can generate 3072-dimensional vectors, and text-embedding-v3 can generate 1024-dimensional vectors).
For each layer, we generate 200 random queries and compute the top-5 nearest neighbors for each query.
Both C-HNSW and Base-HNSW are configured with identical parameters: $M = 32$, $efSearch = 100$, and $efConstruction = 100$.
Importantly, our method maintains comparable retrieval accuracy to Base-HNSW, with recall of 0.5537 and 0.6058, respectively.
We compare our hierarchical search based on C-HNSW with a baseline approach (Base-HNSW), which independently builds a vector index for attributed communities at each layer and performs retrieval separately for each.
As shown in Figure~\ref{fig:c_hnsw-efficiency}, on the large-scale synthetic dataset, C-HNSW achieves up to a 5.4$\times$ speedup (On level 1, C-HNSW takes 1.861 seconds, while Base-HNSW takes 10.125 seconds.) and is on average 3.5$\times$ faster than Base-HNSW. 
Additional details are provided in the appendix.

\pgfplotstableread[row sep=\\,col sep=&]{
lvl & bhnsw & chnsw \\
1 & 10.536 & 3.872 \\
2 & 10.125 & 1.861 \\
3 & 9.505 & 1.83 \\
4 & 8.919 & 1.792 \\
5 & 8.155 & 1.746 \\
6 & 6.856 & 1.653 \\
7 & 5.428 & 1.581 \\
8 & 2.773 & 1.381 \\
9 & 1.135 & 1.092 \\
10 & 0.406 & 0.716 \\
11 & 0.146 & 0.453 \\
12 & 5.816 & 1.635 \\
}\HNSWtime

\begin{figure}[h]
    \centering
    \quad \ref{c_hnsw_index}\\
    \begin{tikzpicture}[scale=0.45]
        \begin{axis}[
            ybar=0.5pt,
            bar width=0.5cm,
            width=\textwidth,
            height=0.25\textwidth,
            xlabel={\Huge \bf Level}, 
            xtick=data,	
            xticklabels={0,1,2,3,4,5,6,7,8,9,10,Avg.},
            legend style={
              at={(0.0,1.05)}, 
anchor=south west, 
legend columns=-1, 
            align=left, 
            draw=none},
            legend image code/.code={
                \draw [#1, line width=0.5pt] (0cm,-0.1cm) rectangle (0.3cm,0.2cm); },
            legend to name=c_hnsw_index,
            xmin=0,xmax=13,
            ymin=0,ymax=15,
            tick align=inside,
            ticklabel style={font=\Huge},
            every axis plot/.append style={line width = 2.5pt},
            every axis/.append style={line width = 2.5pt},
            ylabel={\textbf{\Huge time (s)}}
            ]
\addplot[fill=c6] table[x=lvl,y=bhnsw]{\HNSWtime};
\addplot[fill=c8] table[x=lvl,y=chnsw]{\HNSWtime};
\legend{\small {Base-HNSW}, \small {C-HNSW}}
            \end{axis}
    \end{tikzpicture}
    \caption{C-HNSW and Base-HNSW query efficiency.}
    \label{fig:c_hnsw-efficiency}
\end{figure}

We also conducted experiments on a synthetic dataset of 1024-dimensional vectors, keeping all other parameters unchanged. The results are shown in Figure~\ref{fig:c_hnsw-efficiency-2}.
As 1024 dimensions are more representative of commonly used high-dimensional embeddings, our method still achieves over 5× speedup in the best case and an average speedup of 3× compared to the baseline.

\pgfplotstableread[row sep=\\,col sep=&]{
lvl & bhnsw & chnsw \\
1 & 12.1495 & 3.9887 \\
2 & 10.4776 & 2.0359\\
3 & 9.5378 & 1.9015\\
4 & 7.9879 & 1.8319\\
5 & 7.2649 & 1.7477 \\
6 & 5.9064 & 1.6665 \\
7 & 4.6033 & 1.5740 \\
8 & 2.4135 & 1.3755 \\
9 & 1.0340 & 1.1214 \\
10 & 0.4291 & 0.7797 \\
11 & 0.1702 & 0.5701 \\
12 & 6.179 & 1.69 \\
}\HNSWtimethousand

\begin{figure}[h]
    \centering
    \quad \ref{c_hnsw_index_2}\\
    \begin{tikzpicture}[scale=0.45]
        \begin{axis}[
            ybar=0.5pt,
            bar width=0.5cm,
            width=\textwidth,
            height=0.25\textwidth,
            xlabel={\Huge \bf Level}, 
            xtick=data,	
            xticklabels={0,1,2,3,4,5,6,7,8,9,10,Avg.},
            legend style={
              at={(0.0,1.05)}, 
anchor=south west, 
legend columns=-1, 
            align=left, 
            draw=none},
            legend image code/.code={
                \draw [#1, line width=0.5pt] (0cm,-0.1cm) rectangle (0.3cm,0.2cm); },
            legend to name=c_hnsw_index_2,
            xmin=0,xmax=13,
            ymin=0,ymax=15,
            tick align=inside,
            ticklabel style={font=\Huge},
            every axis plot/.append style={line width = 2.5pt},
            every axis/.append style={line width = 2.5pt},
            ylabel={\textbf{\Huge time (s)}}
            ]
\addplot[fill=c6] table[x=lvl,y=bhnsw]{\HNSWtimethousand};
\addplot[fill=c8] table[x=lvl,y=chnsw]{\HNSWtimethousand};
\legend{\small {Base-HNSW}, \small {C-HNSW}}
            \end{axis}
    \end{tikzpicture}
    \caption{C-HNSW and Base-HNSW query efficiency on 1024-dimensional vector dataset.}
    \label{fig:c_hnsw-efficiency-2}
\end{figure}

$\bullet$ \textbf{Efficiency of indexing phrase.} 
Figure \ref{fig:index-efficiency} shows the index construction time and token usage for different methods.
The cost of building an index for {ArchRAG} is similar to that of {GraphRAG}, but due to the need for community summarization, both take a higher time cost and token usage than {HippoRAG}.

\begin{figure}[h]
    \centering
    \quad \ref{eff_index}\\
    \subfigure[Time cost]{
    \begin{tikzpicture}[scale=0.45]
            \begin{axis}[
                ybar=0.5pt,
                bar width=0.6cm,
                width=0.44\textwidth,
                height=0.25\textwidth,
                xtick=data,	
                xticklabels={\huge Multihop-RAG,\huge HotpotQA},
                legend style={
                  at={(0.0,1.05)}, 
    anchor=south west, 
    legend columns=3, 
                align=left, 
                draw=none},
                legend image code/.code={
                    \draw [#1, line width=0.5pt] (0cm,-0.1cm) rectangle (0.3cm,0.2cm); },
                legend to name=eff_index,
                xmin=1,xmax=5,
                ymin=1000,ymax=100000,
                ytick = {1000,10000,100000},
                yticklabels = {$10^{3}$, $10^{4}$, OOT},
                ymode = log,
                tick align=inside,
                ticklabel style={font=\Huge},
                every axis plot/.append style={line width = 2.5pt},
                every axis/.append style={line width = 2.5pt},
                ylabel={\textbf{\huge time (s)}}
                ]
\addplot[fill=c9] table[x=datasets,y=raptor]{\Indextime};
\addplot[fill=c5] table[x=datasets,y=Hippo]{\Indextime};
\addplot[fill=c6] table[x=datasets,y=light]{\Indextime};
\addplot[fill=c7] table[x=datasets,y=GraphRAG]{\Indextime};
\addplot[fill=c8] table[x=datasets,y=ArchRAG]{\Indextime};
\legend{\small {RAPTOR}, \small {HippoRAG},\small {LightRAG},\small {GraphRAG},\small {ArchRAG}}
            \end{axis}
        \end{tikzpicture}
        \label{fig:index-time}
    }
    \subfigure[Token cost]{
		\begin{tikzpicture}[scale=0.45]
            \begin{axis}[
                ybar=0.5pt,
                bar width=0.6cm,
                width=0.44\textwidth,
                height=0.25\textwidth,
                xtick=data,
                xticklabels={\huge Multihop-RAG,\huge HotpotQA},
                xmin=1,xmax=5,
                ymin=1000000,ymax=100000000,
                ytick = {1000000, 10000000, 100000000},
                yticklabels = {1,10,100},
                ymode = log,
                tick align=inside,
                ticklabel style={font=\Huge},
                every axis plot/.append style={line width = 2.5pt},
                every axis/.append style={line width = 2.5pt},
                ylabel={\textbf{\huge token (M)}}
                ]
\draw[c9] (axis cs:3.3,1000000) node[rotate=90, anchor=west] {\Huge \bf OOT};
\addplot[fill=c9] table[x=datasets,y=raptor]{\Indextoken};
\addplot[fill=c5] table[x=datasets,y=Hippo]{\Indextoken};
\addplot[fill=c6] table[x=datasets,y=light]{\Indextoken};
\addplot[fill=c7] table[x=datasets,y=GraphRAG]{\Indextoken};
\addplot[fill=c8] table[x=datasets,y=ArchRAG]{\Indextoken};
            \end{axis}
        \end{tikzpicture}
        \label{fig:index-token}
	}
    \caption{Comparison of indexing efficiency.}
    \label{fig:index-efficiency}
\end{figure}

\begin{table*}[ht]
    \centering
    \caption{Comparing ArchRAG with other RAG methods on the specific QA tasks under different LLM backbone models.}
    \setlength{\tabcolsep}{10pt}
    \begin{tabular}{llcccc}
        \toprule
\multirow{2}{*}{LLM backbone} &\multirow{2}{*}{Methods} &  \multicolumn{2}{c}{Multihop-RAG} & \multicolumn{2}{c}{HotpotQA} \\
\cmidrule(lr){3-4} \cmidrule(lr){5-6}
       & & (Accuracy) & (Recall) & (Accuracy) & (Recall)\\
\midrule
\multirow{4}{*}{Llama3.1-8B} & Vanilla RAG & 58.6 & 31.4 & 50.6 & 56.1 \\
&HippoRAG& 38.9 & 19.1 & \underline{51.3} & \underline{56.8} \\ 
&RAPTOR& \underline{59.1}& \underline{34.1} & N/A & N/A \\
&ArchRAG& \textbf{68.8} & \textbf{37.2} & \textbf{65.4} & \textbf{69.2} \\
\midrule
\multirow{4}{*}{GPT-3.5-turbo} & Vanilla RAG & 65.9 & \underline{32.8} & \underline{60.7} & \textbf{65.9} \\
&HippoRAG& \underline{68.9} & 31.4 & 58.0 & 62.3 \\
&RAPTOR& 64.4 & \textbf{34.6} & N/A & N/A \\
&ArchRAG& \textbf{67.2} & 31.5 & \textbf{62.8} & \underline{65.0} \\
\midrule
\multirow{4}{*}{GPT-4o-mini}&Vanilla RAG& \underline{71.4} & \underline{32.8} & \underline{68.2} & \underline{70.1} \\
&HippoRAG& 70.5 & 31.6 & 65.0 & 68.5 \\
&RAPTOR& 70.1 & 32.6 & N/A & N/A \\
&ArchRAG& \textbf{77.3} & \textbf{33.8} & \textbf{69.9} & \textbf{73.8} \\
        \bottomrule
    \end{tabular}
    \label{tab:backbone}
\end{table*}

To further demonstrate the effectiveness of ArchRAG, we conduct the following experiments:

$\bullet$ \textbf{Effectiveness of LLM backbones.}
Given the limited budget, we restrict our evaluation to GPT-4o-mini and GPT-3.5-turbo as the LLM backbones, and compare a representative subset of strong RAG methods on the HotpotQA and Multihop-RAG datasets.
As strong LLMs with hundreds of billions of parameters (e.g., GPT-3.5-turbo) possess enhanced capabilities, our proposed ArchRAG may also benefit from performance improvement.
As shown in Table \ref{tab:backbone}, the results of Llama 3.1-8B are similar to those of GPT-3.5-turbo, as Llama3.1's capabilities are comparable to those of GPT-3.5-turbo \cite{dubey2024llama}.
GPT-4o-mini performs better than other LLM backbones because of its exceptional reasoning capabilities.

Besides, we have compared several strong RAG baselines under different LLM backbones.
As LLMs' parameters and reasoning capabilities increase, all RAG approaches benefit from performance gains, especially HippoRAG.
ArchRAG consistently achieves state-of-the-art performance across most settings.

$\bullet$ \textbf{Community quality of different clustering methods.} 
We evaluate the community quality of our proposed LLM-based hierarchical clustering framework using four clustering algorithms (weighted Leiden, Spectral Clustering~\cite{von2007tutorial}, SCAN~\cite{xu2007scan}, and node2vec~\cite{grover2016node2vec} with KMeans), combined with two graph augmentation techniques (the KNN algorithm and CODICIL~\cite{CODICIL2013efficient}).
The resulting communities are assessed using the Calinski-Harabasz Index (CHI)~\cite{calinski1974CHIndex} and Cosine Similarity (Sim) \cite{charikar2002similarity}, where higher values indicate better quality.
%
Further details on the clustering implementation and evaluation metrics can be found in the appendix.
Figures \ref{fig:ch} and \ref{fig:sim} show that combining KNN or CODICIL with the weighted Leiden algorithm significantly enhances community detection quality.

\pgfplotstableread[row sep=\\,col sep=&]{
methods & cos & topk  \\ 
2 & 0.8836 & 0.8772 \\
4 & 0.8549 & 0.7791 \\
6 & 0.6963 & 0.7813 \\
8 & 0.7689 & 0.7662 \\
}\simhotpot
\pgfplotstableread[row sep=\\,col sep=&]{
methods & cos & topk  \\ 
2 & 4.8161 & 4.9554 \\
4 & 4.6785 & 1.3950 \\
6 & 1.1961 & 1.5548 \\
8 & 1.0184 & 1.0041 \\
}\chhotpot

\pgfplotstableread[row sep=\\,col sep=&]{
methods & cos & topk  \\ 
2 & 0.8869 & 0.8737 \\
4 & 0.8441 & 0.7856 \\
6 & 0.6941 & 0.7748 \\
8 & 0.7653 & 0.4793 \\
}\simmultihop
\pgfplotstableread[row sep=\\,col sep=&]{
methods & cos & topk  \\ 
2 & 4.6783 & 4.5866 \\
4 & 4.7229 & 1.5135 \\
6 & 1.4620 & 1.4930 \\
8 & 1.0127 & 1.0021 \\
}\chmultihop

\begin{figure}[h]
    \centering
    \quad \ref{quality_ch}\\
    \subfigure[Multihop-RAG]{
    \begin{tikzpicture}[scale=0.45]
            \begin{axis}[
                ybar=0.5pt,
                bar width=0.6cm,
                width=0.48\textwidth,
                height=0.25\textwidth,
                xtick=data,	
                xticklabels={{\fontsize{12pt}{14pt}\selectfont Leiden}, {\fontsize{12pt}{14pt}\selectfont Spectral}, {\fontsize{12pt}{14pt}\selectfont SCAN}, {\fontsize{12pt}{14pt}\selectfont Node2Vec}},
                legend style={
                anchor=north,legend columns=4,
                draw=none},
                legend image code/.code={
                    \draw [#1, line width=0.5pt] (0cm,-0.1cm) rectangle (0.3cm,0.2cm); },
                legend to name=quality_ch,
                xmin=1,xmax=9,
                ymin=0,ymax=5.5,
                tick align=inside,
                ticklabel style={font=\Huge},
                every axis plot/.append style={line width = 2.5pt},
                every axis/.append style={line width = 2.5pt},
                ylabel={\textbf{{\Huge CHI}}}
                ]
\addplot[fill=c6] table[x=methods,y=cos]{\chmultihop};
\addplot[fill=c3] table[x=methods,y=topk]{\chmultihop};
\legend{\small {KNN},\small {CODICIL}}
            \end{axis}
        \end{tikzpicture}
        \label{fig:ch-multihop}
    }
    \subfigure[HotpotQA]{
		\begin{tikzpicture}[scale=0.45]
            \begin{axis}[
                ybar=0.5pt,
                bar width=0.6cm,
                width=0.48\textwidth,
                height=0.25\textwidth,
                xtick=data,	
                xticklabels={{\fontsize{12pt}{14pt}\selectfont Leiden}, {\fontsize{12pt}{14pt}\selectfont Spectral}, {\fontsize{12pt}{14pt}\selectfont SCAN}, {\fontsize{12pt}{14pt}\selectfont Node2Vec}},
                legend style={
                anchor=north,legend columns=4,
                draw=none},
                legend image code/.code={
                    \draw [#1, line width=0.5pt] (0cm,-0.1cm) rectangle (0.3cm,0.2cm); },
                xmin=1,xmax=9,
                ymin=0,ymax=5.5,
                tick align=inside,
                ticklabel style={font=\Huge},
                every axis plot/.append style={line width = 2.5pt},
                every axis/.append style={line width = 2.5pt},
                ylabel={\textbf{{\Huge CHI}}}
                ]
\addplot[fill=c6] table[x=methods,y=cos]{\chhotpot};
\addplot[fill=c3] table[x=methods,y=topk]{\chhotpot};
            \end{axis}
        \end{tikzpicture}
        \label{fig:ch-hotpotqa}
	}
    \caption{Community quality evaluated by CH Index.}
    \label{fig:ch}
\end{figure}

\begin{figure}[h]
    \centering
    \quad \ref{quality_sim}\\
    \subfigure[Multihop-RAG]{
    \begin{tikzpicture}[scale=0.45]
            \begin{axis}[
                ybar=0.5pt,
                bar width=0.6cm,
                width=0.47\textwidth,
                height=0.25\textwidth,
                xtick=data,	
                xticklabels={{\fontsize{12pt}{14pt}\selectfont Leiden}, {\fontsize{12pt}{14pt}\selectfont Spectral}, {\fontsize{12pt}{14pt}\selectfont SCAN}, {\fontsize{12pt}{14pt}\selectfont Node2Vec}},
                legend style={
                anchor=north,legend columns=4,
                draw=none},
                legend image code/.code={
                    \draw [#1, line width=0.5pt] (0cm,-0.1cm) rectangle (0.3cm,0.2cm); },
                legend to name=quality_sim,
                xmin=1,xmax=9,
                ymin=0.4,ymax=1,
                tick align=inside,
                ticklabel style={font=\Huge},
                every axis plot/.append style={line width = 2.5pt},
                every axis/.append style={line width = 2.5pt},
                ylabel={\textbf{{\Huge Sim}}}
                ]
\addplot[fill=c6] table[x=methods,y=cos]{\simmultihop};
\addplot[fill=c3] table[x=methods,y=topk]{\simmultihop};
\legend{\small {KNN},\small {CODICIL}}
            \end{axis}
        \end{tikzpicture}
        \label{fig:sim-multihop}
    }
    \subfigure[HotpotQA]{
		\begin{tikzpicture}[scale=0.45]
            \begin{axis}[
                 ybar=0.5pt,
                bar width=0.6cm,
                width=0.47\textwidth,
                height=0.25\textwidth,
                xtick=data,	
                xticklabels={{\fontsize{12pt}{14pt}\selectfont Leiden}, {\fontsize{12pt}{14pt}\selectfont Spectral}, {\fontsize{12pt}{14pt}\selectfont SCAN}, {\fontsize{12pt}{14pt}\selectfont Node2Vec}},
                legend style={
                anchor=north,legend columns=4,
                draw=none},
                legend image code/.code={
                    \draw [#1, line width=0.5pt] (0cm,-0.1cm) rectangle (0.3cm,0.2cm); },
                xmin=1,xmax=9,
                ymin=0.6,ymax=1,
                ytick={0.6,0.8,1},
                tick align=inside,
                ticklabel style={font=\Huge},
                every axis plot/.append style={line width = 2.5pt},
                every axis/.append style={line width = 2.5pt},
                ylabel={\textbf{{\Huge Sim}}}
                ]
\addplot[fill=c6] table[x=methods,y=cos]{\simhotpot};
\addplot[fill=c3] table[x=methods,y=topk]{\simhotpot};

            \end{axis}
        \end{tikzpicture}
        \label{fig:sim-hotpotqa}
	}
    \caption{Community quality evaluated by Cosine Similarity.}
    \label{fig:sim}
\end{figure}

$\bullet$ \textbf{Effectiveness of our attributed clustering algorithm.}
To further demonstrate the effectiveness of our attributed clustering algorithm, we evaluate the quality of communities (i.e., CHI and Cosine Similarity) generated by our attributed clustering algorithm compared to those produced by the Leiden algorithm, which is used in GraphRAG for structural clustering.
As shown in Table~\ref{tab:c-quality}, our attribute-based clustering consistently yields higher-quality communities.

\begin{table*}[ht]
    \centering
    \setlength{\tabcolsep}{13pt}
    \caption{Comparing ArchRAG with other RAG methods on the RAG-QA Area dataset. Each entry denotes the win ratio and win + tie ratio of the corresponding method against the ground-truth annotations, based on LLM evaluation. }
    \begin{tabular}{c|ccccc}
\toprule
Method& Lifestyle & Recreation & Science & Technology & Writing  \\
\midrule
Vanilla RAG & 17.5 / 20.5 & 17.0 / 25.0&32.5 / 37.0&28.5 / 34.0&15.0 / 16.5\\
HippoRAG & 26.5 / 26.5& 29.5 / 30.5&49.5 / 49.5&42.0 / 42.5&21.0 / 21.5 \\
RAPTOR & 17.0 / 19.5& 18.5 / 24.5&39.0 / 45.0&33.0 / 35.5&25.0 / 26.5 \\
ArchRAG & \textbf{49.5 / 50.0}& \textbf{41.5 / 41.5} & \textbf{56.0 / 56.0} & \textbf{59.0 / 59.5} & \textbf{45.0 / 45.0} \\
\bottomrule
    \end{tabular}
    \label{tab:area}
\end{table*}

\begin{table}[h]
    \centering
    \setlength{\tabcolsep}{10pt}
    \caption{Comparison of Community Quality between Our Attributed Clustering Method and Leiden}
    \begin{tabular}{c|cccc}
\toprule
\multirow{2}{*}{Method} &  \multicolumn{2}{c}{Multihop-RAG} & \multicolumn{2}{c}{HotpotQA} \\
\cmidrule(lr){2-3} \cmidrule(lr){4-5}
&  (CHI) & (Sim) & (CHI) & (Sim)\\
\midrule
 Leiden & 3.02 & 0.71& 3.42 & 0.71 \\
 Ours & 4.68 & 0.89 &4.82& 0.88 \\
\bottomrule
    \end{tabular}
    \label{tab:c-quality}
\end{table}

$\bullet$ \textbf{More experiments on the additional dataset.}
We further conduct experiments on the RAG-QA Arena dataset~\cite{han2024rag}, a high-quality, multi-domain benchmark featuring human-annotated, coherent long-form answers.
To the best of our capability, we use publicly available data from five domains (including lifestyle, recreation, science, technology, and writing), selecting 200 questions per domain.
Following prior work, we employ LLMs as evaluators to compare the RAG-generated responses with ground-truth answers in terms of win ratio and win + tie ratio.
We evaluate the top-performing methods (including Vanilla RAG, HippoRAG, RAPTOR, and our ArchRAG) and present the results in terms of the win ratio and win + tie ratio against the ground-truth annotations in the table below.
As shown in Table~\ref{tab:area}, ArchRAG consistently achieves state-of-the-art performance across all evaluated settings.

%

\pgfplotstableread[row sep=\\,col sep=&]{
topk & Accuracy & Recall \\
1 & 63.2 & 37 \\
3 & 66.1 & 36.5 \\
5 & 68.8 & 37.2 \\
7 & 63.4 & 37.1 \\
9 & 67.4 & 37 \\
}\topkmultihop
\pgfplotstableread[row sep=\\,col sep=&]{
topk & Accuracy & Recall \\
1 & 61.5 & 65.6 \\
3 & 61.2 & 65 \\
5 & 65.4 & 69.2 \\
7 & 63.6 & 67.3 \\
9 & 47.7 & 52.5 \\
}\topkhotpot

\begin{figure}[h]
    \centering

    \quad \ref{topk_leg} \\
    \vspace{-5pt}
    \subfigure[Multihop-RAG]{
		\begin{tikzpicture}[scale=0.45]
			\begin{axis}[
                legend style = {
                    legend columns=-1,
                    inner sep=0pt,
                    draw=none,
                },
legend image post style={scale=1.1, line width=1pt},
                legend to name=topk_leg,
                xmin=-1, xmax=11,
				ymin=20, ymax=80,
				xtick = {1,3,5,7,9},
				xticklabels = {2,3,5,7,9},
				mark size=6.0pt, 
				width=0.45\textwidth,
                height=0.25\textwidth,
ylabel={\Huge \bf Metrics},
ticklabel style={font=\Huge},
every axis plot/.append style={line width = 2.5pt},
every axis/.append style={line width = 2.5pt},
				]
\addplot [mark=triangle,color=c6] table[x=topk,y=Accuracy]{\topkmultihop};
\addplot [mark=o,color=c3] table[x=topk,y=Recall]{\topkmultihop};
\legend{\small Accuracy,\small Recall}
			\end{axis}
		\end{tikzpicture}
	}
    \subfigure[HotpotQA]{
		\begin{tikzpicture}[scale=0.45]
			\begin{axis}[
                xmin=-1, xmax=11,
				ymin=20, ymax=80,
				xtick = {1,3,5,7,9},
				xticklabels = {2,3,5,7,9},
				mark size=6.0pt, 
				width=0.45\textwidth,
                height=0.25\textwidth,
ylabel={\Huge \bf Metrics},
ticklabel style={font=\Huge},
every axis plot/.append style={line width = 2.5pt},
every axis/.append style={line width = 2.5pt},
				]
\addplot [mark=triangle,color=c6] table[x=topk,y=Accuracy]{\topkhotpot};
\addplot [mark=o,color=c3] table[x=topk,y=Recall]{\topkhotpot};
			\end{axis}
		\end{tikzpicture}
	}
    \caption{Comparative analysis of the different numbers of retrieval elements in ArchRAG.}
    \label{fig:topk}
\end{figure}

$\bullet$ \textbf{Effect of $k$ values.}
We compare the performance of ArchRAG under different retrieved elements.
As shown in Figure \ref{fig:topk}, the performance of ArchRAG shows little variation when selecting different retrieval elements (i.e., communities and entities in each layer). 
This suggests that the adaptive filtering process can reliably extract the most relevant information from the retrieval elements and integrate it to generate the answer.



$\bullet$ \textbf{Case study. }
We present an additional case study from Multihop-RAG.
As shown in Figure~\ref{fig:case}, only our method generates the correct answer, while others either provide incorrect or irrelevant information. We only show the core output for brevity, with the remaining marked as ``</>''.
We also show the retrieval and adaptive filtering process of ArchRAG in Figure~\ref{fig:case-retrieval}. The results demonstrate that ArchRAG effectively retrieves relevant information and filters out noise, leading to a correct final answer.

\begin{table}[h]
    \centering
    \caption{Distribution of HippoRAG's ER Errors}
    \begin{tabular}{c|c|c}
\toprule
Datasets & Null Entity Rate& \makecell{Low-Quality \\ Entity Rate} \\
\midrule
Multihop-RAG & 1.3\% & 11.9\% \\
HotpotQA & 5.0\% & 15.8\% \\
\bottomrule
    \end{tabular}
    \label{tab:wrong}
\end{table}

$\bullet$ \textbf{Discussion of the performance of other graph-based RAG methods. }
On the Multihop-RAG dataset, {HippoRAG} performs worse than retrieval-only methods while outperforming retrieval-based methods on the HotpotQA dataset.
This is mainly because, on the HotpotQA dataset, passages are segmented by the expert annotators; that is, passages can provide more concise information, whereas on  Multihop-RAG, passages are segmented based on chunk size, which may cause the LLM to lose context and produce incorrect answers.
Besides, HippoRAG also suffers from inaccurate entity recognition.
In addition, we also provide a brief analysis of why some graph-based RAG methods underperform.
Specifically, they often fail to retrieve relevant information during the search process.
For example:

\begin{itemize}
    \item \textit{GraphRAG:} The suboptimal performance of GraphRAG has been widely observed in research.
    For instance, Zhang et al.~\cite{zhang2024sirerag} find that GraphRAG tends to respond with "I don't know" when retrieving irrelevant content, indicating that it prefers not to give a concrete answer.
    PIKE-RAG~\cite{wang2025pike} identifies that GraphRAG tends to echo the query and include meta-information about answers within its graph structure.
    Through manual verification, we confirm the presence of similar issues in our experimental results. In particular, GraphRAG-Local Search also frequently returns irrelevant content.
    \item \textit{HippoRAG/LightRAG: } We attribute part of HippoRAG's suboptimal performance to entity recognition issues. HippoRAG first extracts relevant entities from questions and then retrieves chunks based on these entities.
    However, it sometimes fails to identify appropriate entities (Entity Recognition Errors, ER Errors), and we conduct corresponding experiments, as shown in Table~\ref{tab:wrong}. For example, on HotpotQA, HippoRAG fails to detect entities in 5.0\%  of questions and identifies low-quality ones in 15.8\% (e.g., less than 3 characters, numeric-only terms, or sentence entities with more than 5 words). 
    These issues make it difficult to match accurate entities in the KG, resulting in irrelevant retrieval and suboptimal performance. Similar issues also impair LightRAG's performance because it requires accurate topic identification from questions. We also observe that as the LLM's capability improves, the performance of these methods increases significantly, as shown in the experiment \textbf{Effectiveness of LLM backbone}.
\end{itemize}

\begin{figure*}[htbp]
    \centering
    \renewcommand{\arraystretch}{1.2}

    \vspace{0.5em}
    \begin{tcolorbox}[colback=gray!5, colframe=gray!50, title=Question \& Ground-truth]
\textit{Question:} Who is the individual associated with generative AI technology that was reportedly ousted from a leading AI organization, but is recognized for brilliance and generosity in the industry, and is currently planning to launch a new venture according to reports from TechCrunch and Fortune?

\textit{Ground-truth:} Sam Altman
    \end{tcolorbox}
\vspace{-0.5em}

    \begin{tcolorbox}[ragstyle, title=Responses of Different Methods]

\textbf{Zero-shot} \\
{\bf \textcolor{red}{Jesse Schell}}. \newline
He is the CTO of Calm and was the lead designer on the ``Star Wars" theme park ride at Disney.

\vspace{1.3em}\hrule\vspace{0.8em}

\textbf{Vanilla RAG} \\
It seems like there are multiple articles and topics discussed in this text. Here's a brief summary of each section: \newline
1.{\bf \textcolor{red}{Robotics predictions.}} The article discusses the future of robotics, with predictions made by Brian Heater. $\cdots$ \newline
2.{\bf \textcolor{red}{TechCrunch's Week in Review.}} This section summarizes various tech industry happenings. $\cdots$ \newline
3.{\bf \textcolor{red}{Bumble gets a new CEO.}} Bumble announces a change in leadership. $\cdots$

\vspace{1.3em}\hrule\vspace{0.8em}

\textbf{HippoRAG} \\
There is {\bf \textcolor{red}{no question} }in the provided text. It appears to be a collection of news articles and updates related to OpenAI's ChatGPT. $\cdots$

\vspace{1.3em}\hrule\vspace{0.8em}

\textbf{GraphRAG-Global} \\
\textit{Key Points and Implications} \newline
The individual associated with generative AI technology who was reportedly ousted is {\bf \textcolor{red}{Andrew Ng}}. $\cdots$ \newline
According to multiple analysts, {\bf \textcolor{red}{Andrew Ng}} [Data: Reports (5, 6)] $\cdots$ \newline
\textit{Implications of the New Venture} $\cdots$

\vspace{1.3em}\hrule\vspace{0.8em}

\textbf{ArchRAG} \\
{\bf \textcolor{cyan}{Sam Altman}}.  \newline
This information is supported by multiple analysts' reports, which highlight {\bf \textcolor{cyan}{Sam Altman's}} involvement in various AI-related projects. The direct answer is simply ``{\bf \textcolor{cyan}{Sam Altman}}''. $\cdots$ \newline
Note: I have not included any information from Analyst 2, 3, or 7 as their reports are not relevant. $\cdots$
\end{tcolorbox}
\caption{Case study of responses by different RAG methods on a question from the Multihop-RAG dataset. {\bf \textcolor{red}{RED TEXT}} indicates incorrect content, {\bf \textcolor{cyan}{CYAN TEXT}} denotes correct content, and $\cdots$ marks omitted irrelevant parts.}
    \label{fig:case}
\end{figure*}

\begin{figure*}
    \centering  
    \vspace{2em}
    \begin{lstlisting}
Top Retrieved Entities:
 - Entity 1058: OPENAI
   Description: OpenAI is a tech company founded by Sam Altman ...

 - Community 4097: Microsoft, OpenAI, and AI Regulation
   Summary: Discussion between Satya Nadella and Sam Altman ...

Adaptive Filtering Result:
 - Sam Altman is among the backers of an AI startup.
 - Score: 80.0 (Reports: 1) ...
\end{lstlisting}
    \vspace{-1em}

    \caption{ArchRAG Retrieval \& Filtering Output. $...$ marks omitted irrelevant parts.}
    \label{fig:case-retrieval}
\end{figure*}

\begin{figure*}[t] 
\begin{AIbox}{Prompt for generating abstract questions}
{\bf Prompt:} \\
{
Given the following description of a dataset:

\{description\}

Please identify 5 potential users who would engage with this dataset. For each user, list 5 tasks they would perform with this dataset. Then, for each (user, task) combination, generate 5 questions that require a high-level understanding of the entire dataset.

Output the results in the following structure:
\begin{description}
    \item[- User 1: \text{[user description]}]
    \begin{description}
        \item \item[- Task 1: \text{[task description]}]
        \begin{description}  
            \item \item  - Question 1:
            \item  - Question 2:
            \item  - Question 3:
            \item  - Question 4:
            \item  - Question 5:
        \end{description}  
       \item[- Task 2: \text{[task description]}]
            \item ...
       \item[- Task 5: \text{[task description]}]
    \end{description}   
\end{description}    

\begin{description}
    \item[- User 2: \text{[user description]}]
        \item...
\end{description}

\begin{description}
    \item[- User 5: \text{[user description]}]
        \item...
\end{description} 
    
    Note that there are 5 users and 5 tasks for each user, resulting in 25 tasks in total. Each task should have 5 questions, resulting in 125 questions in total.
    The Output should present the whole tasks and questions for each user.
    
    Output:

}

\end{AIbox} 
\caption{The prompt for generating abstract questions.}
\label{fig:prompt_summary}
\end{figure*}

\begin{figure*}[t] 
\begin{AIbox}{Prompt for LLM-based multi-dimensional comparison}
{\bf Prompt:} \\
{
You will evaluate two answers to the same question based on three criteria: {\bf Comprehensiveness}, {\bf Diversity}, {\bf Empowerment}, and {\bf Directness}.
\begin{itemize}
    \item Comprehensiveness: How much detail does the answer provide to cover all aspects and details of the question?
    \item Diversity: How varied and rich is the answer in providing different perspectives and insights on the question?
    \item Empowerment: How well does the answer help the reader understand and make informed judgments about the topic?
    \item Directness: How specifically and clearly does the answer address the question?
\end{itemize}

For each criterion, choose the better answer (either Answer 1 or Answer 2) and explain why. Then, select an overall winner based on these four categories.

Here is the {\bf question}: 
\begin{verbatim}
    Question: {query}
\end{verbatim}

Here are the two answers:
\begin{verbatim}
    Answer 1: {answer1}
    Answer 2: {answer2}
\end{verbatim}

Evaluate both answers using the four criteria listed above and provide detailed explanations for each criterion. Output your evaluation in the following JSON format:
\begin{verbatim}
{
    "Comprehensiveness": {
        "Winner": "[Answer 1 or Answer 2]",
        "Explanation": "[Provide one sentence explanation here]"
    },
    "Diversity": {
        "Winner": "[Answer 1 or Answer 2]",
        "Explanation": "[Provide one sentence explanation here]"
    },
    "Empowerment": {
        "Winner": "[Answer 1 or Answer 2]",
        "Explanation": "[Provide one sentence explanation here]"
    },
    "Overall Winner": {
        "Winner": "[Answer 1 or Answer 2]",
        "Explanation": "[Briefly summarize why this answer is the overall
        winner]"
    }
}
\end{verbatim}

Output:

}

\end{AIbox} 
\caption{The prompt for the evaluation of abstract QA.}
\label{fig:eval_summary}
\end{figure*}

\begin{figure*}[t] 
\begin{AIbox}{Prompt for Adaptive filtering-based generation}
{\bf Filter Prompt:} \\
{
\# Role

You are a helpful assistant responding to questions about data in the tables provided.

\# Goal

Generate a response consisting of a list of key points that respond to the user's question, summarizing all relevant information in the input data tables.
You should use the data provided in the data tables below as the primary context for generating the response.
If you don't know the answer or if the input data tables do not contain sufficient information to provide an answer, just say so. Do not make anything up.

Each key point in the response should have the following element:
\begin{itemize}
    \item Description: A comprehensive description of the point.
    \item Importance Score: An integer score between 0-100 that indicates how important the point is in answering the user's question. An `I don't know' type of response should have a score of 0.
\end{itemize}

The response should be JSON formatted as follows:

\begin{verbatim}
{
"points": [
    {
        "description": "Description of point 1", 
        "score": score_value
    },
    // ... more points
]
}
\end{verbatim}
\# User Question
\begin{verbatim}
    {user_query}
\end{verbatim}

\# Data tables
\begin{verbatim}
    {context_data}
\end{verbatim}
Output:
}
\tcbline
{\bf Merge Prompt:} \\
{
\# Role

You are a helpful assistant responding to questions and may use the provided data as a reference.

\# Goal

You should incorporate insights from all the reports from multiple analysts who focused on different parts of the dataset to support your answer. Please note that the provided information may contain inaccuracies or be unrelated. If the provided information does not address the question, please respond using what you know:

\begin{itemize}
    \item A response that utilizes the provided information, ensuring that all irrelevant details from the analysts' reports are removed.
    \item A response to the user's query based on your existing knowledge when <Analyst Reports> is empty.
\end{itemize}

The final response should merge the relevant information into a comprehensive answer that clearly explains all key points and implications, tailored to the appropriate response length and format.
Note that the analysts' reports provided below are ranked in the {\it descending order of importance}. Do not include information where the supporting evidence for it is not provided.

\# Target response length and format 

\begin{verbatim}
    {response_format}
\end{verbatim}

\# User Question 

\begin{verbatim}
    {user_query}
\end{verbatim}

\# Analyst Reports 

\begin{verbatim}
    {report_data}
\end{verbatim}

Output:
}

\end{AIbox} 
\caption{The prompt for adaptive filtering-based generation.}
\label{fig:prompt_afg}
\end{figure*}

\clearpage